\documentclass[11pt,letterpaper]{article}

\usepackage{amsmath,amsthm,amssymb,latexsym,amsfonts,mathabx,amscd,bbm,a4wide}
\usepackage[top=1.3in,bottom=1.3in,left=1in,right=1in]{geometry}
\usepackage{dsfont}
\usepackage{mathrsfs}
\usepackage{setspace}
\usepackage{hyperref}
\usepackage{mathtools} 
\usepackage{enumerate}
\usepackage{latexsym,amsfonts,bbm, mathabx}
\newcommand{\ud}{\, \mathrm{d}}
\newcommand{\E}{\mathds{E}}
\newcommand{\llargest}{l_{N,>}^{(1),\omega}}
\newcommand{\Ilargest}{I_{N,>}^{(1),\omega}}

\newcommand{\llargestzwei}{l_{N,>}^{(2),\omega}}
\newcommand{\Ilargestzwei}{I_{N,>}^{(2),\omega}}
\newcommand{\llargestN}[1]{l_{N,>}^{(#1),\omega}}

\newcommand{\IlargestN}[1]{I_{N,>}^{(#1),\omega}}
\newcommand{\tildeIlargestN}[1]{\widetilde{I}_{N,>}^{(#1),\omega}}
\newcommand{\llargestk}[1]{\tilde l_{k,>}^{(#1),\omega}}

\newcommand{\e}{\mathrm{e}}
\newcommand{\laplace}{\mathop{}\!\mathbin\bigtriangleup}

\newcommand{\cutofffunctioneins}[2]{\chi_{#1, #2}}

\newcommand{\WWS}{\mathcal S}

\newcommand{\mugeneral}{\hat{\mu}}
\newcommand{\PoissonrandommeasureRomega}{\mathcal M_{\nu}^{\omega}}
\newcommand{\PoissonrandommeasureR}{\mathcal M_{\nu}}

\newcommand{\RP}{V}
\newcommand{\RPN}{\RP_N}
\newcommand{\CSU}{\mathcal C_{u}}
\newcommand{\CSUright}{\CSU^{\mathrm{right}}}
\newcommand{\CSUleft}{\CSU^{\mathrm{left}}}

\newcommand{\constantENNleNinfty}{C_1}
\numberwithin{equation}{section}
\newtheorem{theorem}{Theorem}[section]
\newtheorem{lemma}[theorem]{Lemma}
\newtheorem{prop}[theorem]{Proposition}
\newtheorem{cor}[theorem]{Corollary}
\newtheorem{remark}[theorem]{Remark}

\theoremstyle{definition}
\newtheorem{definition}[theorem]{Definition}

\allowdisplaybreaks

\begin{document}

    \frenchspacing

    \allowdisplaybreaks[2]

    \thispagestyle{empty}

    \vspace*{1cm}

    \begin{center}
	
        {\Large \bf On Bose--Einstein condensation in one-dimensional noninteracting Bose gases in the presence of soft Poisson obstacles} \\

        \vspace*{2cm}
	
        {\large Maximilian~Pechmann \footnote{E-mail address: {\tt mpechmann@utk.edu}}}
	
        \vspace*{5mm}
		
        Department of Mathematics\\
        University of Tennessee\\
        Knoxville, TN 37996\\
        USA
	
        \vspace*{5mm}
	
    \end{center}

    \begin{abstract}
        \noindent We study Bose--Einstein condensation (BEC) in one-dimensional noninteracting Bose gases in Poisson random potentials on $\mathds R$ with single-site potentials that are nonnegative, compactly supported, and bounded measurable functions in the grand-canonical ensemble at positive temperatures in the thermodynamic limit. For particle densities that are larger than a critical one, we prove the following: With arbitrarily high probability when choosing the fixed strength of the random potential sufficiently large, BEC where only the ground state is macroscopically occupied occurs. If the strength of the Poisson random potential converges to infinity in a certain sense but arbitrarily slowly, then this kind of BEC occurs in probability and in the $r$th mean, $r \ge 1$. Furthermore, in Poisson random potentials of any fixed strength an arbitrarily high probability for type-I g-BEC is also obtained by allowing sufficiently many one-particle states to be macroscopically occupied.
    \end{abstract}

    \section{Introduction}
        (Conventional) \emph{Bose--Einstein condensation} (BEC) is a macroscopic occupation of a one-particle state and occurs, under certain circumstances, in bosonic particle systems. In the case of noninteracting Bose gases (bosonic particle systems without interparticle interaction), a necessary but not sufficient requirement for the occurrence of BEC is the presence of \emph{generalized Bose--Einstein condensation} (g-BEC). This broader definition only requires a macroscopic occupation of an arbitrarily small energy band of one-particle states \cite{LanWil79,van1982generalized, van1983condensation, van1986general2, van1986general, lenoble2004bose}. Depending on the quantity of macroscopically occupied one-particle states in the condensate one then distinguishes three types: Type-I g-BEC is said to occur if the number of macroscopically occupied one-particle states is finite but at least one. If there are infinitely many macroscopically occupied one-particle states, the condensation is said to be of type II. Lastly, a generalized condensate in which none of the one-particle states are macroscopically occupied is called a type-III g-BEC. Showing the occurrence of g-BEC is easier than the occurrence of BEC and involves verifying that a certain critical density is finite as a main step. Proving BEC or, similarly, determining the type of g-BEC, however, seems to require fairly accurate knowledge about the gaps between the eigenvalues of the corresponding one-particle (random) Schrödinger operator at the bottom of the spectrum \cite{van1982generalized,KPSConditionTypeOneBEC2020}, which is often difficult to obtain. Note that the definition of type-I g-BEC is more restrictive than our definition of BEC.
 
        Random potentials are known to be able to trigger and enhance the occurrence of g-BEC in noninteracting Bose gases, see, e.g., \cite{lenoble2004bose} and \cite[Appendix A]{KPS182}. Thus, the study of Bose gases in such potentials is of great interest. This holds especially true for Poisson random potentials as they are commonly used to model systems with structural disorder.
        Although it is believed that repulsive interactions between the particles eventually need to be taken into account, exploring noninteracting Bose gases with respect to BEC is nevertheless an important first step and of independent interest \cite{lenoble2004bose} (see also \cite{stolz1995localization}, \cite{germinet2005localization}, \cite{klein2007localization}, \cite{seiringer2016decay}).

        The \emph{Kac--Luttinger conjecture} presumes that g-BEC in noninteracting Bose gases in Poisson random potentials that have compactly supported, nonnegative measurable functions as their single-site potentials is generally of type I or II, that is, BEC occurs \cite{kac1973bose, kac1974bose, LenobleZagrebnovLuttingerSy}. To the best of our knowledge, however, the type of g-BEC in random potentials at positive temperatures has been so far rigorously determined only for one-dimensional Poisson random potentials whose single-site potentials consists of the Dirac delta function $\delta$. The \emph{Luttinger--Sy model} \cite{luttinger1973bose, luttinger1973low} has a Poisson random potential on $\mathds R$ with, informally, a single-site potential of the form $\gamma \delta$ where $\gamma = \infty$, that is, one has Dirichlet boundary conditions at all atoms of each realization of the Poisson random measure. This model is easier to explore, because the singularity of this random potential eliminates quantum tunneling effects \cite[p. 3]{jaeck2010nature}. It has been proved that in this Luttinger--Sy model a type-I g-BEC, where only the ground state of the corresponding one-particle random Schrödinger operator is macroscopically occupied, occurs in probability and in the $r$th mean, $r \ge 1$, in \cite{KPSConditionTypeOneBEC2020}, and in a slightly different setting $\mathds P$-almost surely \cite{LenobleZagrebnovLuttingerSy}, if and only if the particle density is larger than a critical density. In addition, it has been shown that in the Luttinger--Sy model with finite interaction strength, that is, in the case of a Poisson random potential on $\mathds R$ with, informally, a single-site potential of the form $\gamma \delta$ with $\gamma > 0$, a type-I g-BEC occurs with probability arbitrarily close to one for particle densities larger than a critical density as long as one allows sufficiently many one-particle states to be macroscopically occupied \cite{KPS182}. Despite their singularities, Poisson random potentials on $\mathds R$ with such single-site potentials, and in particular the infinite potential strength of the Luttinger--Sy model, are believed to be good approximations with respect to the occurrence of BEC for noninteracting Bose gases in more realistic Poisson random potentials on $\mathds R$, such as ones that have nonnegative, bounded functions as their single-site potentials \cite[p. 8]{LenobleZagrebnovLuttingerSy}, \cite[p. 14]{VeniaminovKlopp14}.

        In this work, we study one-dimensional noninteracting Bose gases in Poisson random potentials on $\mathds R$ with \emph{soft obstacles}, that is, with single-site potentials that are nonnegative, compactly supported, and bounded measurable functions with respect to the occurrence of BEC in the thermodynamic limit and in the grand-canonical ensemble at positive temperatures. For this model, we confirm the Kac--Luttinger conjecture in the following sense. Under the assumption that the particle density is larger than a finite critical density, we prove: A type-I g-BEC in which only the ground state is macroscopically occupied occurs with arbitrarily high probability if the random potential has a, in a certain sense, sufficiently large strength. The probability for this kind of condensation converges to one and, consequently, such a type-I g-BEC occurs in probability and in the $r$th mean, $r \ge 1$, if the strength of the Poisson random potential converges in a certain sense but arbitrarily slowly to infinity in the thermodynamic limit. One also obtains an arbitrarily high probability for type-I g-BEC in the case of a Poisson random potential of any fixed strength when allowing sufficiently many one-particle states to be macroscopically occupied. As a side note we mention that the same results hold true for the Luttinger--Sy model with finite interaction strength, and we thus confirm and extend the results in \cite{KPS182} while using a different, more direct method.

        The outline of this paper is as follows. In Section~\ref{SecPrelim}, we introduce our model and collect some of its well-known properties that we use in the rest of this work. We also define the different kind of condensations that we have mentioned so far. Next, we prove important, lesser known facts regarding the Poisson point process on $\mathds R$ in Section~\ref{Porperties of the Poisson Point process}. In Section~\ref{section energy bounds} we derive an upper bound for the first and lower bounds for higher eigenvalues of the corresponding one-particle random Schrödinger operator. Using facts from Section~\ref{Porperties of the Poisson Point process} and~\ref{section energy bounds}, we formulate and prove our main results, namely Corollaries~\ref{Corollary 1},~\ref{Corollary 2}, and~\ref{Corollary 3}, in Section~\ref{secMainResults}. Lastly, in the Appendix~\ref{secAppendix} we provide details regarding the almost sure occurrence of g-BEC if one considers a Poisson random potential with a strength that converges to infinity.

    \section{Preliminaries and model}\label{SecPrelim}

        Let $(\Omega, \mathscr A, \mathds P)$ be a probability space and $\nu > 0$ a constant. We firstly state the definition of a Poisson random potential $V$ on $\mathds R$ \cite{pastur1992spectra, leschke2003survey}.

        A random measure $\PoissonrandommeasureR$ on $\mathds R$ is called a Poisson random measure on $\mathds R$ if the following two points hold: For any $n \in \mathds N$ and any pairwise disjoint Borel sets $B_1, \ldots, B_n \subset \mathds R$ the random variables $\PoissonrandommeasureR(B_1), \ldots, \PoissonrandommeasureR(B_n)$ are independent. For any bounded Borel set $B \subset \mathds R$ and any $m \in \mathds N_0$, the random variable $\PoissonrandommeasureR(B)$ has the distribution
        \begin{equation}
            \mathds P \big( \PoissonrandommeasureRomega(B) = m \big) = \e^{- \nu \lambda(B)} \dfrac{ (\nu \lambda(B))^m}{m!} \ ,
        \end{equation}
        $\lambda$ being the Lebesgue-measure on $\mathds R$.
	
        Let $\PoissonrandommeasureR$ be a Poisson random measure on $\mathds R$. Suppose that $u : \mathds R \to \mathds R$ is a nonrandom function that is nonzero on a nonempty and open subset of $\mathds R$ and meets the Birman--Solomyak requirement $\sum_{m \in \mathds Z} ( \, \int_{[m - 1/2, m + 1/2]} |u(x)|^2 \, \mathrm{d} x )^{1/2} < \infty$. Then we call
        \begin{equation} \label{def RP}
            V : \Omega \times \mathds R \to \mathds R, \ (\omega, x) \mapsto V(\omega, x) := \int\limits_{\mathds R} u(x-y) \, \, \PoissonrandommeasureRomega( \mathrm{d} y)
        \end{equation}
        a \emph{Poisson random potential} on $\mathds R$ and the function $u$ its \emph{single-site potential}.
        A Poisson random potential on $\mathds R$ is metrically transitive. The constant $\nu$ is called the rate of the Poisson random potential. We call a Poisson random potential \emph{positive} if $u \ge 0$.
 
        In this work, we consider Poisson random potentials $V$ on $\mathds R$ with single-site potentials $u$ that are nonnegative and bounded measurable functions with compact support $[-\CSUleft, \CSUright]$ where $\CSUleft, \CSUright > 0$. Furthermore, we define
        \begin{align}
            \CSU & := \CSUright + \CSU^{\text{left}} \label{Def CSU} \ .
        \end{align}
        We also define the \emph{strength} of $V$ as
        \begin{align}
            \WWS := \min\Big\{ \int_{0}^{\CSUright} u(x) \ud x, \int_{-\CSUleft}^{0} u(x) \ud x \Big\} \ .
        \end{align}
 
        In addition, we consider sequences of Poisson random potentials $(V_N)_{N \in \mathds N}$ on $\mathds R$,
        \begin{equation} \label{def RP converges to infty}
            V_N : \Omega \times \mathds R \to \mathds R, \ (\omega, x) \mapsto V_N(\omega, x) := \int\limits_{\mathds R} u_N(x-y) \, \, \PoissonrandommeasureRomega( \mathrm{d} y)
        \end{equation}
        for all $N \in \mathds N$, where
        \begin{align} \label{def RP converges to infty zwei}
            u_N(x) := \WWS_N u(x)
        \end{align}
        for all $N \in \mathds N$ and all $x \in \mathds R$. We assume that $u$ is a nonnegative, bounded measurable function with compact support $[-\CSUleft, \CSUright]$, $\CSUleft, \CSUright > 0$ and that $(\WWS_N)_{N \in \mathds N} \subset (0,\infty) $ is monotonically increasing and converges to infinity.
        We call $u$ the single-site potential and $(\WWS_N)_{N \in \mathds N}$ the strength of $(V_N)_{N \in \mathds N}$. 
  
        \begin{remark} \label{remark RP}
            We refer to the case \eqref{def RP} as a \emph{Poisson random potential of fixed strength} and the case \eqref{def RP converges to infty} as a \emph{Poisson random potential with a strength that converges to infinity}. To simplify the notation, we also write $(V_N)_{N \in \mathds N}$ for a Poisson random potential $V$ with a fixed strength $\WWS$ and single-site potential $u$, and mean $V_N := V$ and $u_N := u$ for all $\omega \in \Omega$ and $N \in \mathds N$ in this case. In either case, we assume that the single-site potential $u$ is a nonnegative, compactly supported, and bounded measurable function throughout this work.
        \end{remark}

        We study one-dimensional noninteracting Bose gases in (sequences of) positive Poisson random potentials $(V_N)_{N \in \mathds N}$ on $\mathds R$ in the grand-canonical ensemble at an arbitrary, fixed, positive temperature $T>0$ in the thermodynamic limit. Regardless of whether we assume a Poisson random potential of fixed strength or a Poisson random potential with a strength that converges to infinity, we consequently consider a sequence of systems in which $N \in \mathds N$ bosons are confined in the intervals
        \begin{equation}
            \Lambda_N = (-L_N/2,L_N/2)
        \end{equation}
        with $N/L_N = \rho$, where $\rho > 0$ is a constant and called the \emph{particle density}. We impose \emph{Dirichlet boundary conditions} on $\Lambda_N$. Note that when exploring BEC in noninteracting Bose gases, it is sufficient to obtain properties of the corresponding sequence of one-particle random Schrödinger operators
        \begin{align}\label{Definition RSO}
            H_{N,\RPN}: \Omega \to \{\text{linear operators on } L^2(\Lambda_N)\} , \ \omega \mapsto H_{N,\RPN}^{\omega} := (-\laplace + V_N(\omega))_{\Lambda_N}^{\text{D}} \ ,
        \end{align}
        $N \in \mathds N$, see, e.g., Proposition~\ref{Proposition 5.3}, \cite{KPSConditionTypeOneBEC2020}, or \cite{lenoble2004bose}.
        
        Here, $H_{N,\RPN}^{\omega}$ is the linear operator on $\mathrm{L}^2(\Lambda_N$) that is uniquely defined as the self-adjoint extension of the operator
        \begin{align}
            \mathrm{L}^2(\Lambda_N) \to \mathrm{L}^2(\Lambda_N), \quad \psi(x) \mapsto - \dfrac{\mathrm{d}^2}{\mathrm{d} x^2} \psi(x) + V_N(\omega,x) \psi(x)
        \end{align}
        with domain $C_c^{\infty}(\Lambda_N)$, for $\mathds P$-almost all $\omega \in \Omega$ and all $N \in \mathds N$ \cite[Theorem 5.1]{pastur1992spectra}. $\mathds P$-almost surely and for all $N \in \mathds N$, properties of the operator $H_{N,\RPN}^{\omega}$ include having a purely discrete spectrum, that is, a spectrum that consists only of isolated eigenvalues of finite multiplicities. We denote the sequence of these eigenvalues, written in ascending order and each eigenvalue repeated according to its multiplicity, by $(E_{N,\RPN}^{j,\omega})_{j \in \mathds N} \subset (0,\infty)$, and the associated sequence of eigenfunctions by $(\varphi_{N,\RPN}^{j,\omega})_{j \in \mathds N}$.
        Due to the purely discrete spectrum, we can define the function
        \begin{align}
            \mathcal N_{N,\RPN}^{\mathrm{I}, \omega} : \mathds R \to [0,\infty), \quad E \mapsto \mathcal N_{N,\RPN}^{\mathrm{I},\omega}(E) := | \{ j \in \mathds N : E_{N,\RPN}^{j,\omega} < E \}|\ ,
        \end{align}
        which is called the integrated density of states and is a measure-defining function. By $\mathcal N_{N,\RPN}^{\omega}$ we denote the measure that is uniquely defined by the condition $\mathcal N_{N,\RPN}^{\omega}([a,b))= \mathcal N_{N,\RPN}^{\mathrm{I}, \omega}(b) - \mathcal N_{N,\RPN}^{\mathrm{I}, \omega}(a)$ for all $a < b$. We call $\mathcal N_{N,\RPN}^{\omega}$ the \emph{density of states}.

        In the case of a Poisson random potential $V$ of fixed strength, the sequence $(\mathcal N_{N,V}^{\omega})_{N \in \mathds N}$ $\mathds P$-almost surely converges in the vague sense to a nonrandom measure $\mathcal N_{\infty,V}$, which we call the limiting density of states. The nonrandom function
        \begin{align}
            \mathcal N_{\infty,V}^{\mathrm{I}} : \mathds R \to [0,\infty), \quad E \mapsto \mathcal N_{\infty,V}^{\mathrm{I}}(E) := \mathcal N_{\infty,V} ((-\infty,E))
        \end{align}
        is called the limiting integrated density of states. The limiting integrated density of states has a \emph{Lifshitz tails property}, and, in particular, obeys \cite[Theorem 10.2]{pastur1992spectra}, \cite[Theorem 4.1]{leschke2003survey}
        \begin{align} \label{Lifshitz tail}
            \mathcal N_{\infty,V}^{\mathrm{I}}(E) = \exp \left( -\nu \pi E^{-1/2} \big[1 + o(1) \big] \right) \qquad \text{ for } E \searrow 0 \ .
        \end{align}
        Moreover, there is a constant $\constantENNleNinfty > 0$ and a constant $\widetilde E > 0$ such that for all $0 < E < \widetilde E$ and all $N \in \mathds N$
        \begin{align} \label{ENN le c Ninfty}
            \E[\mathcal N_{N,\RP}^{\mathrm{I},\omega}(E) ] \le \constantENNleNinfty \mathcal N_{\infty,\RP}^{\mathrm{I}}(E) \ ,
        \end{align}
        and if $(\mathcal N_{N,\RP}^{\mathrm{I},\omega}(E))_{N \in \mathds N}$ $\mathds P$-almost surely converges to $\mathcal N_{\infty,\RP}^{\mathrm{I}}(E)$ for all $E \in \mathds R$ then one can choose $\constantENNleNinfty = 1$, see \cite[Theorem 5.25]{pastur1992spectra} or, for more details, \cite[Theorem 4.4.1]{pechmanndiss}.

        Given a system of $N \in \mathds N$ noninteracting bosons described by a one-particle random Schrödinger operator $H_{N,\RPN}$, the (mean) \emph{occupation number} $n_{N,\RPN}^{j,\omega}$ of a one-particle eigenstate $\varphi^{j,\omega}_{N,\RPN}$ of $H_{N,\RPN}^{\omega}$ with corresponding eigenvalue $E_{N,\RPN}^{j,\omega}$, that is, the number of particles occupying $\varphi^{j,\omega}_{N,\RPN}$ is, in the grand-canonical ensemble at inverse temperature $\beta := (k_B T)^{-1} \in (0,\infty)$, $k_B>0$ being the Boltzmann constant, given by
        \begin{align} \label{occupation number}
            n_{N,\RPN}^{j,\omega} := \Big( \mathrm{e}^{\beta (E_{N,\RPN}^{j,\omega} - \mu_{N,\RPN}^{\omega})} - 1 \Big)^{-1}
        \end{align}
        for $\mathds{P}$-almost all $\omega \in \Omega$. The chemical potential $\mu_{N,\RPN}^{\omega} \in (-\infty,E_{N,\RPN}^{1,\omega})$ is, for $\mathds{P}$-almost all $\omega \in \Omega$ and for all $N \in \mathds{N}$, uniquely determined by the equation
        \begin{align} \label{Gleichung Bedingung mu aequivalent}
            \int\limits_{\mathds{R}} \Big( \mathrm{e}^{\beta (E - \mu_{N,\RPN}^{\omega})} - 1 \Big)^{-1} \, \mathcal N_{N,\RPN}^{\omega}(\mathrm{d} E) =\rho \ .
        \end{align}
        Note that $\int_{\mathds{R}} ( \mathrm{e}^{\beta (E - \mu_{N,\RPN}^{\omega})} - 1 )^{-1} \, \mathcal N_{N,\RPN}^{\omega}(\mathrm{d} E) = L_N^{-1} \sum_{j \in \mathds N} n_{N,\RPN}^{j,\omega}$. Also, $n_N^{i,\omega} \ge n_N^{j,\omega}$ whenever $i \le j$.

        The \emph{critical density} of a noninteracting Bose gas in a Poisson random potential $V$ of fixed strength is defined as
        \begin{align} \label{Definition critical density}
            \rho_{c,\RP} := \int\limits_{\mathds R} \mathcal{B}(E) \, \mathcal N_{\infty,V}(\mathrm{d} E) \ ,
        \end{align}
        which is a finite number. Here, we have introduced the function
        \begin{align}
            \mathcal{B}: \mathds R \to \mathds [0,\infty), \ E \mapsto \mathcal{B}(E) := \big( \mathrm{e}^{\beta E} - 1 \big)^{-1} \, \mathds 1_{(0,\infty)}(E)
        \end{align}
        in order to simplify the notation. Moreover, we write
        \begin{align} \label{def rho0}
            \rho_{0,V} := \max\{\rho - \rho_{c,\RP},0\}
        \end{align}
        for convenience.

        For $\mathds P$-almost all $\omega \in \Omega$ the atoms $\{ \hat x_j^{\omega}\}_j$ of $\PoissonrandommeasureRomega$ can be labeled by $\mathds Z$ and such that
        \begin{equation}
            \ldots < \hat x_{-1}^{\omega} < \hat x_{0}^{\omega} < 0 < \hat x_1^{\omega} < \hat x_2^{\omega} < \ldots \ .
        \end{equation}
        By $\kappa_N^{\omega}$ we refer to the number of atoms of $\PoissonrandommeasureRomega$ within $\Lambda_N$,
        \begin{align}
            \kappa_N^{\omega} := | \{ j \in \mathds Z : \hat x_j^{\omega} \in \Lambda_N\}| \ .
        \end{align}
        We $\mathds P$-almost surely have
        \begin{align} \label{kappa LN to nu}
            \lim\limits_{N \to \infty} \dfrac{\kappa_N^{\omega}}{L_N} = \nu \ .
        \end{align}
        Moreover,
        \begin{align} \label{kappaNomega between 1 minus epsilon L and 1 plus epsilon L}
            \lim\limits_{N \to \infty} \mathds P \big( (1-\varepsilon) \nu L_N \le \kappa_N^{\omega} \le (1 + \varepsilon) \nu L_N \big) = 1 \ ,
        \end{align}
        for all $\varepsilon > 0$, see, e.g., \cite[Section 3.3.2]{SeiYngZag12}.
    
        The set $\{ \hat l^{j}\}_{j\in \mathds Z \backslash\{0\}}$ where
        \begin{align}\label{definition hat lj}
        \hat l^{j} : \Omega \to \mathds R, \ \omega \mapsto \hat l^{j,\omega} := |(\hat x_j^{\omega}, \hat x_{j+1}^{\omega})|
        \end{align}
        is a set of independent and identically distributed random variables with common probability density $\nu \e^{-\nu l} \mathds 1_{(0,\infty)}(l)$, $l \in \mathds R$. We define
        \begin{align}\label{definition INj}
            I_N^{j,\omega} := (\hat x_{j}^{\omega}, \hat x_{j+1}^{\omega}) \cap \Lambda_N
        \end{align}
        and
        \begin{align} \label{definition lNj}
            l_N^{j,\omega} := |I_N^{j,\omega}|
        \end{align}
        for all $j \in \mathds Z$, all $N \in \mathds N$, and $\mathds P$-almost all $\omega \in \Omega$. For all $N \in \mathds N$, there are $\mathds P$-almost surely only finitely many atoms of $\PoissonrandommeasureRomega$ within $\Lambda_N$ and, consequently, only finitely many nonempty intervals in $(I_N^{j,\omega})_{j\in \mathds Z}$. For each $N \in \mathds N$, we order these intervals according to their lengths, and denote the largest interval by $\Ilargest$, the second largest by $I_{N,>}^{(2),\omega}$, etc. We write $\llargestN{j}$ for the length of $I_{N,>}^{(j),\omega}$, $j \in \mathds Z$.

        \begin{remark} \label{Remark typical omega}
            When we write $\widetilde \Omega$ in the rest of this work, then we mean the subset of $\Omega$ with measure $\mathds P(\widetilde \Omega) = 1$ such that for all $\omega \in \widetilde \Omega$ the $\mathds P$-almost sure properties we have mentioned so far hold. 
        \end{remark}

        Next, we introduce the relevant definitions of and an important fact regarding BEC. We consider a system of noninteracting bosons described by a sequence of one-particle random Schrödinger operators $(H_{N,\RPN})_{N \in \mathds N}$.

        \begin{definition} \label{Definition gBEC}
            We say that \emph{generalized BEC (g-BEC) $\mathds P$-almost surely occurs} if and only if
            \begin{align*}
                \mathds P \Big( \lim\limits_{\varepsilon \searrow 0} \limsup\limits_{N \to \infty} \dfrac{1}{N} \sum\limits_{j \in \mathds N : E_{N,\RPN}^{j,\omega} \le \varepsilon} n_{N,\RPN}^{j,\omega} > 0 \Big) = 1\ .
            \end{align*}
        \end{definition}

        \begin{theorem}\label{Theorem g-BEC}
            Let $V$ be a Poisson random potential of fixed strength. Then the following statements hold: If $\rho>\rho_{c,\RP}$, then
            \begin{align}
                \mathds P \Bigg( \lim\limits_{\varepsilon \searrow 0} \liminf\limits_{N \to \infty} \dfrac{1}{N} \sum\limits_{j \in \mathds N : E_{N,\RP}^{j,\omega} \le \varepsilon} n_{N,\RP}^{j,\omega} = \dfrac{\rho_{0,V}}{\rho} \Bigg) = 1 
            \end{align}
            and in particular g-BEC $\mathds P$-almost surely occurs. If $\rho\le\rho_{c,\RP}$, then
            \begin{align}
                \mathds P \Bigg( \lim\limits_{\varepsilon \searrow 0} \limsup\limits_{N \to \infty} \dfrac{1}{N} \sum\limits_{j \in \mathds N : E_{N,\RP}^{j,\omega} \le \varepsilon} n_{N,\RP}^{j,\omega} = 0 \Bigg) = 1 \ ,
            \end{align}
            that is, g-BEC does $\mathds P$-almost surely not occur.
    
            If $(V_N)_{N \in \mathds N}$ is a Poisson random potential with a strength that converges to infinity, then
            \begin{align}
                \mathds P \Bigg( \lim\limits_{\varepsilon \searrow 0} \liminf\limits_{N \to \infty} \dfrac{1}{N} \sum\limits_{j \in \mathds N : E_{N,\RPN}^{j,\omega} \le \varepsilon} n_{N,\RPN}^{j,\omega} \ge \dfrac{\rho_{0,V_1}}{\rho} \Bigg) = 1  \ .
            \end{align}
            In particular, g-BEC $\mathds P$-almost surely occurs if $\rho > \rho_{c,\RP_1}$.
        \end{theorem}

        A proof of the first part of Theorem~\ref{Theorem g-BEC} can be found in \cite{lenoble2004bose} in combination with \cite[Appendix A]{KPS182}. We give a proof of the second part of Theorem~\ref{Theorem g-BEC} in the appendix, see Theorem~\ref{Theorem A.7}.

        In Section~\ref{secMainResults}, we show that, in several probabilistic senses depending on the strength of the Poisson random potential, a type-I g-BEC occurs and consequently that finitely many eigenstates of the corresponding sequence of one-particle random Schrödinger operators are macroscopically occupied and all other eigenstates are not macroscopically occupied. By stating the definition of a ``not macroscopically occupied'' one-particle eigenstate, we simultaneously give the definition of a ``macroscopically occupied'' one-particle eigenstate.

        \begin{definition}[Macroscopic occupation]\label{Definition makroskopische Besetzung}
            The $j$th eigenstate of $(H_{N,\RPN})_{N \in \mathds N}$ is said to be \emph{not macroscopically occupied $\mathds{P}$-almost surely / in the $r$th mean / in probability} if and only if $n_{N,\RPN}^{j,\omega}/N$ converges to zero $\mathds P$-almost surely / in the $r$th mean / in probability. In addition, we say that the $j$th eigenstate of $(H_{N,\RPN}^{\omega})_{N \in \mathds N}$ is \emph{macroscopically occupied with probability almost one} if and only if for all $0 < \varepsilon < 1$ there exists a constant $c> 0$ such that $\liminf_{N \to \infty} \mathds P(n_{N,\RPN}^{j,\omega} / N \ge c) \ge 1 - \varepsilon$.
        \end{definition}
	
        We note that we have stated the definitions of being not macroscopically occupied $\mathds{P}$-almost surely, in the $r$th mean, and in probability stronger and consequently the definition of being macroscopically occupied $\mathds{P}$-almost surely, in the $r$th mean, and in probability  weaker than possible. However, we will show that if the strength of the Poisson random potential converges in a certain sense but arbitrarily slowly to infinity, then $(n_{N,\RPN}^{1,\omega}/N)_{N \in \mathds N}$ converges to a positive value in probability and in the $r$th mean, see Corollary~\ref{Corollary 2}.
 
        \begin{definition}[Type-I g-BEC] \label{Definition types of BEC}
            We speak of \emph{type-I g-BEC} in the probabilistic sense $\mathds P$-almost surely / in the $r$th mean / in probability / with probability almost one if and only if the number of eigenstates of $(H_{N,\RPN})_{N \in \mathds N}$ that are macroscopically occupied $\mathds{P}$-almost surely / in the $r$th mean / in probability / with probability almost one is finite but at least one.
        \end{definition}
  
        Furthermore, type-II g-BEC in the probabilistic sense $\mathds P$-almost surely / in the $r$th mean / in probability / with probability almost one is said to occur if and only if there are infinitely many eigenstates of $H_{N,\RPN}^{\omega}$ macroscopically occupied $\mathds{P}$-almost surely / in the $r$th mean / in probability / with probability almost one. If and only if g-BEC $\mathds P$-almost surely occurs but none of the eigenstates are macroscopically occupied $\mathds{P}$-almost surely / in the $r$th mean / in probability / with probability almost one, then one may say that type-III g-BEC occurs in the probabilistic sense $\mathds P$-almost surely / in the $r$th mean / in probability / with probability almost one.
  
        Note that while a $\mathds{P}$-almost sure macroscopic occupation of an eigenstate of $(H_{N,\RPN})_{N \in \mathds N}$ implies the $\mathds{P}$-almost sure occurrence of g-BEC, the converse is in general not true, see, e.g., \cite{van1982generalized} or \cite{LenobleZagrebnovLuttingerSy}. One can deduce sufficient conditions for the occurrence of the types of BEC in the various probabilistic senses with Definitions~\ref{Definition makroskopische Besetzung} and~\ref{Definition types of BEC}. For the convenience of the reader and since type-I g-BEC will be the main concern in this paper, we state the following sufficient conditions. Type-I g-BEC occurs in probability if there exists a $\hat c \in \mathds N$ and a $c > 0$ such that
        \begin{align}
            \lim\limits_{N \to \infty} \mathds P\left( \left| \sum\limits_{i=1}^{\hat c} \dfrac{n_{N,\RPN}^{i,\omega}}{N} - c \right| < \eta , \dfrac{n_{N,\RPN}^{\hat c + 1,\omega}}{N} < \eta' \right) = 1
        \end{align}
        for all $\eta, \eta' > 0$. If there exists an $r \ge 1$, a $c > 0$, and a $\hat c \in \mathds N$ such that
        \begin{align}
            \lim\limits_{N \to \infty} \E \, \Big| \sum\limits_{i=1}^{\hat c} \dfrac{n_{N,\RPN}^{i,\omega}}{N} - c \Big|^r = 0 \ \text{ and } \ \lim\limits_{N \to \infty} \E \Big( \dfrac{n_{N,\RPN}^{\hat c + 1,\omega}}{N} \Big)^r = 0 \ ,
        \end{align}
        then type-I g-BEC occurs in the $r$th mean. Lastly, a type-I g-BEC occurs with probability almost one if for all $0 < \varepsilon < 1$ there exist a $\hat c \in \mathds N$ and a $c > 0$ such that for all $\eta > 0$,
        \begin{align}
            \liminf\limits_{N \to \infty} \mathds P \left( \sum\limits_{i = 1}^{\hat c} \dfrac{n_{N,\RPN}^{i,\omega}}{N} \ge c , \dfrac{n_{N,\RPN}^{\hat c + 1,\omega}}{N} < \eta \right) \ge 1 - \varepsilon \ .
        \end{align}

        As we will see in Section~\ref{secMainResults}, to show the occurrence of type-I g-BEC in our model it is sufficient to prove that with high probability a certain energy gap between the lowest eigenvalue $E_{N,\RPN}^{1,\omega}$ and some higher eigenvalue $E_{N,\RPN}^{j,\omega}$, $2 \le j \in \mathds N$, of the one-particle random Schrödinger operator $H_{N,\RPN} = (-\laplace + V_N)_{\Lambda_N}^{\mathrm{D}}$ is fulfilled, see also \cite{KPSConditionTypeOneBEC2020}. Hence, we introduce, for any $0 < \zeta_2 < \zeta_1 < 1$ and any $2 \le j \in \mathds N$, the event
        \begin{align} \label{Def omega N j}
            \Omega_{N,\RPN}^{j,\zeta_1,\zeta_2} & := \left\{ \omega \in \widetilde \Omega : E_{N,\RPN}^{j, \omega} - E_{N,\RPN}^{1,\omega} \ge N^{-1+\zeta_1}, E_{N,\RPN}^{1,\omega} \le \left[ \left( 1 + \zeta_2 \right) \dfrac{\nu \pi}{\ln(L_N)} \right]^2 \right\} \ .
        \end{align}
        In order to show that this event occurs with high probability, we study the difference between the lengths $\llargest$ and $\llargestN{j}$, $2 \le j \in \mathds N$, which are generated by a Poisson point process on $\mathds R$, in the next section. Subsequently, we prove an upper bound for $E_{N,\RPN}^{1,\omega}$ as well as as an lower bound for $E_{N,\RPN}^{j, \omega}$, $j \in \mathds N$, that depend on $\llargest$ and on $\llargestN{j}$, respectively, in Section~\ref{section energy bounds}. To conclude this section, we would like to make the following remark as a side note.

        \begin{remark} \label{Remark LS with finite strength}
            Our results in this work also hold in the case of the \emph{Luttinger--Sy model of finite strength}, that is, a one-dimensional noninteracting Bose gas in a Poisson random potential on $\mathds R$ with, informally, a single-site potential
            \begin{align*}
                u : \mathds R \to \mathds R, \ x \mapsto u(x) := \gamma \delta(x)
            \end{align*}
            where $\gamma > 0$ and $\delta$ is the Dirac delta function. Rigorously, the corresponding one-particle Schrödinger operator $H_{N,\gamma}^{\omega}$ is, for all $N \in \mathds N$ and $\mathds P$-almost all $\omega \in \Omega$, the self-adjoint operator on $\mathrm{L}^2(\Lambda_N)$ that is uniquely defined via the quadratic form 
            \begin{equation} \label{quadratic form RSO gammaN}
                \begin{split}
                    q^{\omega}_{N,\gamma}: \mathrm{H}^1_0(\Lambda_N) \times \mathrm{H}^1_0(\Lambda_N) & \to \mathds C \\
                    (\varphi,\psi) & \mapsto \int\limits_{\Lambda_N}\overline{\varphi^{\prime}(x)}\psi^{\prime}(x) \ud x + \gamma \sum_{j \in \mathds Z: \, \hat x_j^{\omega} \in \Lambda_N}\overline{\varphi(\hat x_j^{\omega})} \psi(\hat x_j^{\omega})
                \end{split}
            \end{equation}
            on $\mathrm{L}^2(\Lambda_N)$. See, e.g., \cite[p. 146--149]{pastur1992spectra} and \cite{KPS182} for more details regarding this model. Thus, we confirm and extend the results in \cite{KPS182} with a different, more direct method.  
        \end{remark}

    \section{Properties of the Poisson point process on $\mathds R$} \label{Porperties of the Poisson Point process}

        In addition to the information that we have provided in Section~\ref{SecPrelim}, we collect and prove in this section facts regarding the lengths $(l_N^{j,\omega})_{j \in \mathds Z}$, $l_N^{j,\omega} = |(\hat x_{j-1}^{\omega}, \hat x_j^{\omega} ) \cap \Lambda_N|$, that are generated by a Poisson point process on $\mathds R$ with rate $\nu > 0$. In particular, we are concerned with the difference of $\llargest$ and $\llargestN{j}$, $2 \le j \in \mathds N$. We also refer to \cite{kingman1993poisson} and \cite{last2018lectures} for more information about the Poisson point process. 

        Firstly, we show an improved statement compared to \eqref{kappaNomega between 1 minus epsilon L and 1 plus epsilon L}. Recall that $\kappa_N^{\omega}$ is the number of atoms of a particular realization of the Poisson random measure within $\Lambda_N$.

        \begin{lemma}\label{Lemma 3.1}
            For all $0 < \varepsilon < 1/2$ and for all but finitely many $N \in \mathds N$ we have
            \begin{align*}
                \mathds P \Big( (1-L_N^{-\varepsilon}) L_N \le \kappa_N^{\omega} \le (1+L_N^{-\varepsilon}) L_N \Big) \ge 1 - 2 \e^{ - (\nu/3) L_N^{1-2\varepsilon}} \ .
            \end{align*}
        \end{lemma}
        \begin{proof}
            Let $0 < \varepsilon < 1/2$ be arbitrarily given. For all $N \in \mathds N$ we have
            \begin{align*}
                \mathds{P}\left( \kappa_N^{\omega} \ge (1 + L_N^{-\varepsilon}) \nu L_N \right) & = \sum\limits_{m = (1 + L_N^{-\varepsilon}) \nu L_N}^{\infty} \e^{-\nu L_N} \dfrac{(\nu L_N)^m}{m!} \\
                & \le \sum\limits_{m = (1 + L_N^{-\varepsilon}) \nu L_N}^{\infty} \e^{-\nu L_N} \dfrac{(\nu L_N)^m}{m!} (1 + L_N^{-\varepsilon})^{m - ((1 + L_N^{-\varepsilon}) \nu L_N)} \\
                & \le \e^{\nu L_N[L_N^{-\varepsilon} - (1 + L_N^{-\varepsilon}) \ln (1 + L_N^{-\varepsilon})]} \ ,
            \end{align*}
            see also \cite[Section 3.3.2]{SeiYngZag12}.
            Next, by using the inequality
            \begin{align} \label{Inequality natural log}
                (1+x) \ln(1+x) - x \ge \dfrac{x^2}{\frac{2}{3} x + 2}
            \end{align}
            for all $x > -1$, we obtain
            \begin{align*}
                \mathds{P}\left( \kappa_N^{\omega} \ge (1 + L_N^{-\varepsilon}) \nu L_N \right) \le \e^{- \nu L_N L_N^{-2\varepsilon} (\tfrac{2}{3} L_N^{-\varepsilon} + 2)^{-1}}
            \end{align*}
            for all $N \in \mathds N$.
            Similarly, one can show that 
            \begin{align*}
                \mathds{P}\left( \kappa_N^{\omega} \le (1 - L_N^{-\varepsilon}) \nu L_N \right) \le \e^{- \nu L_N L_N^{-2\varepsilon} (-\tfrac{2}{3} L_N^{-\varepsilon} + 2)^{-1}}
            \end{align*}
            for all $N \in \mathds N$ such that $L_N^{-\varepsilon} < 1$.
        \end{proof}

        For the convenience of the reader, we mention that inequality \eqref{Inequality natural log} can be shown as follows. We set $f(x) := (1+x) \ln(1+x) - x - x^2/(\frac{2}{3} x + 2)$, $x > -1$.
        Since $f''(x) \ge 0$ for all $x > -1$ and $f'(0) = 0$, one concludes that $f'(x) \ge 0$ for all $x \ge 0$ and $f'(x) \le 0$ for all $-1 < x \le 0$. Furthermore, $f(0)=0$. Thus, $f(x)\ge 0$ for all $x > -1$.

        \begin{lemma} \label{Lemma 3.2}
            For all $\zeta > 0$,
            \begin{align*}
                \lim\limits_{N \to \infty} \mathds P \left( (1-\zeta) \nu^{-1} \ln(L_N) \le \llargest \le (1+\zeta) \nu^{-1} \ln(L_N) \right) = 1 \ .
            \end{align*}
        \end{lemma}
        \begin{proof}
            Let $\zeta > 0$ be arbitrarily given. Since $\mathds P$-almost surely $\lim_{N \to \infty} \kappa_N^{\omega} / L_N = \nu$ and $(\hat l^{j})_{j \in \mathds Z \backslash \{0\}}$ are independent and identically distributed random variables with common probability density $\nu \e^{-\nu l} \mathds 1_{(0,\infty)}(l)$, $l \in \mathds R$, we have
            \begin{align*}
                \limsup\limits_{N \to \infty} \mathds P \left( \llargest > (1+\zeta) \nu^{-1} \ln(L_N) \right) \le \limsup\limits_{N \to \infty} \left( 1 - (1 - L_N^{-1-\zeta})^{(1+\varepsilon)\nu L_N} \right)
            \end{align*}
            for any $\varepsilon > 0$. By using the inequality $1-x^{-1} \le \ln(x) \le x-1$ for all $x > 0$, one shows that this probability converges to zero in the limit $N \to \infty$.
            Similarly, one has
            \begin{align*}
                \limsup\limits_{N \to \infty} \mathds P \left( \llargest < (1-\zeta) \nu^{-1} \ln(L_N) \right) \le \limsup\limits_{N \to \infty} \big( 1 - L_N^{-1+\zeta} \big)^{(1-\varepsilon) \nu L_N}
            \end{align*}
            for any $0 < \varepsilon < 1$ where the right-hand side converges to zero in the limit $N \to \infty$ as well. For more details, see also \cite[Chapter 3, Lemma 3.4]{sznitman1998brownian} or \cite[Lemma A.5]{KPSConditionTypeOneBEC2020}.
        \end{proof}
        
        Let $( \tilde l^{j})_{j \in \mathds N}$ be a sequence of independent and identically distributed random variables with common probability density $\nu \e^{-\nu l} \mathds 1_{(0,\infty)}(l)$, $l \in \mathds R$. For $\mathds P$-almost all $\omega \in \Omega$, we define $\llargestk{1}$ and $\llargestk{2}$ as the largest and the second largest element of the set $\{ \tilde l^{j,\omega}\}_{j=1}^{k}$, respectively. Also, recall definitions~\eqref{definition hat lj} and~\eqref{definition lNj}.

        \begin{lemma} \label{Lemma 3.3}
            For any bounded sequence $(c_N)_{N \in \mathds N} \subset [0,\infty)$ and any $0 < \varepsilon < 1/2$ we have
            \begin{align*}
                & \liminf\limits_{N \to \infty} \mathds P\Big( \llargest - \llargestzwei > c_N \Big) \ge \liminf\limits_{N \to \infty} \mathds P \Big( \tilde l_{\lfloor (1 - L_N^{-\varepsilon}) \nu L_N \rfloor - 2,>}^{(1),\omega} - \tilde l_{\lfloor (1 - L_N^{-\varepsilon}) \nu L_N \rfloor - 2,>}^{(2),\omega} > c_N \Big) \ .
            \end{align*}
        \end{lemma}
        \begin{proof}
            Let a bounded sequence $(c_N)_{N \in \mathds N} \subset [0,\infty)$ and an $0 < \varepsilon < 1/2$ be arbitrarily given.
            For each $\omega \in \widetilde \Omega$ and $N \in \mathds N$, let $\kappa_N^{(L),\omega}$ and $\kappa_N^{(R),\omega}$ be the number of atoms of $\PoissonrandommeasureRomega$ within $(-L_N/2,0)$ and within $(0,L_N/2)$, respectively. Proceeding almost identically as in the proof of Lemma~\ref{Lemma 3.1} one can show that
            \begin{align*}
                \lim_{N \to \infty} \mathds P \Big( (1-L_N^{-\varepsilon})L_N/2 \le \kappa_N^{(L),\omega} \le (1+L_N^{-\varepsilon}) L_N/2 \Big) = 1
            \end{align*}
            and 
            \begin{align*}
                \lim_{N \to \infty} \mathds P \Big( (1-L_N^{-\varepsilon})L_N/2 \le \kappa_N^{(R),\omega} \le (1+L_N^{-\varepsilon}) L_N/2 \Big) = 1 \ .
            \end{align*}
            For convenience, we define the set 
            \begin{align*}
                J_k:= \{ -k, -k+1, \ldots, k-1, k \}\backslash \{0\}, \quad k \in \mathds N \ .
            \end{align*}
            Using that $-x/(1-x) \le \ln(1-x) \le -x$ for all $0<x<1$ we have 
            \begin{align*}
                & \lim\limits_{N \to \infty} \mathds P \Big( \max\big\{ \hat l^{j,\omega} : j \in J_{\lceil (1 + L_N^{-\varepsilon})L_N/2 \rceil }  \backslash J_{\lfloor (1 - L_N^{-\varepsilon})L_N/2 \rfloor - 2} \big\} \le (1 - \varepsilon/2) \nu^{-1} \ln(L_N) \Big) \\
                & \qquad = \lim\limits_{N \to \infty} (1 - L_N^{-1+\varepsilon/2})^{2\nu L_N^{1-\varepsilon}} = 1
            \end{align*}
            and
            \begin{align*}
                & \lim\limits_{N \to \infty} \mathds P \Big(\max\big\{ \hat l^{j,\omega} : j \in J_{\lfloor (1 - L_N^{-\varepsilon})L_N/2 \rfloor - 2} \big\} \ge (1-\varepsilon/3)\nu^{-1} \ln(L_N) \Big)\\
                & \qquad = 1 -  \lim\limits_{N \to \infty} (1 - L_N^{-1+\varepsilon/3})^{(1-L_N^{-\varepsilon}) \nu L_N -2} = 1 \ .
            \end{align*}
            Note that the two outer intervals within $\Lambda_N$ and the interval that contains the zero are relatively small, that is, 
            \begin{align*}
                \lim\limits_{N \to \infty} \mathds P(\Omega_N^{(\text{R})} \cap \Omega_N^{(\text{L})} \cap \Omega_N^{(0)}) = 1
            \end{align*}
            where, for all $N \in \mathds N$ and all $\omega \in \widetilde \Omega$, we have set
            \begin{align*}
                \Omega_N^{(\text{R})} & := \{ \omega \in \widetilde \Omega : |I_N^{\text{R},\omega}| \le (2\nu)^{-1} \ln (L_N) \} \ , \\
                \Omega_N^{(\text{L})} & := \{ \omega \in \widetilde \Omega : |I_N^{\text{L},\omega}| \le (2\nu)^{-1} \ln (L_N) \} \ , \\
                \shortintertext{and}
                \Omega_N^{(0)} & := \{ \omega \in \widetilde \Omega : |I_N^{0,\omega}| \le (2\nu)^{-1} \ln (L_N) \}
            \end{align*}
            with $I_N^{\text{R},\omega} := I_N^{j_N^{\text{max},\omega},\omega}$, $j_N^{\text{max},\omega} := \max\{ j \in \mathds Z: I_N^{j,\omega} \neq \emptyset\}$, and $I_N^{\text{L},\omega} := I_N^{j_N^{\text{min},\omega},\omega}$, $j_N^{\text{min},\omega} := \min\{ j \in \mathds Z : I_N^{j,\omega} \neq \emptyset \}$. Also, for any $\omega \in \widetilde \Omega$ and any $2 \le k \in \mathds N$ the second-largest element of the set $\{ \hat l^{j,\omega} \in J_k\}$ is given by $\max\Big( \big\{ \hat l^{j,\omega} \in J_k \big\} \backslash \big\{ \max ( \{ \hat l^{j,\omega} \in J_k \} ) \big\}   \Big)$.
            
            Since $\mathds P(A \cap B) \ge \mathds P(A) + \mathds P(B) - 1$ 
            for any events $A,B \subset \Omega$, we conclude, for all but finitely many $N \in \mathds N$,
            \begin{align*}
                & \mathds P \Big( \llargest - \llargestzwei > c_N \Big) \\
                & \quad \ge \mathds P \Big( \big\{ \llargest - \llargestzwei > c_N \big\} \cap \Omega_N^{(\text{R})} \cap \Omega_N^{(\text{L})} \cap \Omega_N^{(0)} \\
                & \qquad \qquad \cap \big\{ (1-L_N^{-\varepsilon})L_N/2 \le \kappa_N^{(L),\omega} \le (1+L_N^{-\varepsilon}) L_N/2 \big\} \\
                & \qquad \qquad \cap \big\{ (1-L_N^{-\varepsilon})L_N/2 \le \kappa_N^{(R),\omega} \le (1+L_N^{-\varepsilon}) L_N/2 \big\} \\
                & \qquad \qquad \cap \big\{ \max(\{ \hat l^{j,\omega} : j \in J_{\lceil (1 + L_N^{-\varepsilon})L_N/2 \rceil} \backslash J_{\lfloor(1 - L_N^{-\varepsilon})L_N/2 \rfloor -2} \} ) \le (1 - \varepsilon/2) \nu^{-1} \ln(L_N) \big\} \\  
                & \qquad \qquad \cap \big\{ \max(\{ \hat l^{j,\omega} : j \in J_{\lfloor (1 - L_N^{-\varepsilon})L_N/2 \rfloor - 2} \} ) \ge (1 - \varepsilon/3) \nu^{-1} \ln(L_N) \big\} \Big) \\
                & \quad = \mathds P \Big( \Big\{ \max \big( \{ \hat l^{j,\omega} : j \in J_{\lfloor (1 - L_N^{\varepsilon}) L_N/2 \rfloor - 2} \} \big) \\
                & \qquad \qquad \qquad \qquad - \max \Big( \{ \hat l^{j,\omega} : j \in J_{\lfloor (1 - L_N^{-\varepsilon})L_N/2 \rfloor - 2} \}\backslash \big\{ \max ( \{ \hat l^{j,\omega} : j \in J_{\lfloor (1 - L_N^{-\varepsilon})L_N/2 \rfloor - 2} \} ) \big\} \Big) > c_N \Big\} \\
                & \qquad \qquad \cap \Big\{ \max \Big( \big\{ \hat l^{j,\omega} : j \in J_{\lceil (1 + L_N^{-\varepsilon})L_N/2 \rceil} \big\} \backslash \big\{ \max \big( \{ \hat l^{j,\omega} : j \in J_{\lceil (1 + L_N^{-\varepsilon})L_N/2 \rceil} \} \big) \big\} \Big) \\
                & \qquad \qquad \qquad \in \big\{\hat l^{j,\omega} : j \in J_{\lfloor (1 - L_N^{-\varepsilon})L_N/2 \rfloor - 2} \big\} \Big\} \\
                & \qquad \qquad \cap \Omega_N^{(\text{R})} \cap \Omega_N^{(\text{L})} \cap \Omega_N^{(0)} \\
                & \qquad \qquad \cap \big\{ (1-L_N^{-\varepsilon})L_N/2 \le \kappa_N^{(L),\omega} \le (1+L_N^{-\varepsilon}) L_N/2 \big\} \\
                & \qquad \qquad \cap \big\{ (1-L_N^{-\varepsilon})L_N/2 \le \kappa_N^{(R),\omega} \le (1+L_N^{-\varepsilon}) L_N/2 \big\} \\
                & \qquad \qquad \cap \big \{ \max\{ \hat l^{j,\omega} : j \in J_{\lceil (1 + L_N^{-\varepsilon})L_N/2 \rceil} \backslash J_{\lfloor (1 - L_N^{-\varepsilon})L_N/2 \rfloor - 2} \} \le (1 - \varepsilon/2) \nu^{-1} \ln(L_N) \big\} \\  
                & \qquad \qquad \cap \big \{ \max\{ \hat l^{j,\omega} : j \in J_{\lfloor (1 - L_N^{-\varepsilon})L_N/2 \rfloor - 2} \} \ge (1 - \varepsilon/3) \nu^{-1} \ln(L_N) \big\} \Big) \\
                & \qquad + \mathds P \Big( \Big\{ \max(\{ \hat l^{j,\omega} : j \in J_{\lfloor (1 - L_N^{\varepsilon}) L_N/2 \rfloor - 2} \}) \\
                & \qquad \qquad \qquad \qquad - \max \Big( \big\{ \hat l^{j,\omega} : j \in J_{\lceil (1 + L_N^{-\varepsilon})L_N/2 \rceil} \big \} \backslash \big\{ \max ( \{ \hat l^{j,\omega} : j \in J_{\lfloor (1 - L_N^{-\varepsilon})L_N/2 \rfloor - 2} \} ) \big\} \Big) > c_N \Big\} \\
                & \qquad \qquad \cap \Big\{ \max \Big( \big\{ \hat l^{j,\omega} : j \in J_{\lceil (1 + L_N^{-\varepsilon})L_N/2 \rceil} \big\} \backslash \big\{ \max \big( \{ \hat l^{j,\omega} : j \in J_{\lceil (1 + L_N^{-\varepsilon})L_N/2 \rceil} \} \big) \big\} \Big) \\
                & \qquad \qquad \qquad \in \big\{\hat l^{j,\omega} : j \in J_{\lceil (1+ L_N^{-\varepsilon}) L_N/2 \rceil} \backslash J_{\lfloor (1 - L_N^{-\varepsilon})L_N/2 \rfloor - 2} \big\} \Big\} \\
                & \qquad \qquad \cap \Omega_N^{(\text{R})} \cap \Omega_N^{(\text{L})} \cap \Omega_N^{(0)} \\
                & \qquad \qquad \cap \big\{ (1-L_N^{-\varepsilon})L_N/2 \le \kappa_N^{(L),\omega} \le (1+L_N^{-\varepsilon}) L_N/2 \big\} \\
                & \qquad \qquad \cap \big\{ (1-L_N^{-\varepsilon})L_N/2 \le \kappa_N^{(R),\omega} \le (1+L_N^{-\varepsilon}) L_N/2 \big\} \\
                & \qquad \qquad \cap \big \{ \max\{ \hat l^{j,\omega} : j \in J_{\lceil (1 + L_N^{-\varepsilon})L_N/2 \rceil} \backslash J_{\lfloor (1 - L_N^{-\varepsilon})L_N/2 \rfloor - 2} \} \le (1 - \varepsilon/2) \nu^{-1} \ln(L_N) \big\} \\  
                & \qquad \qquad \cap \big \{ \max\{ \hat l^{j,\omega} : j \in J_{\lfloor (1 - L_N^{-\varepsilon})L_N/2 \rfloor - 2} \} \ge (1 - \varepsilon/3) \nu^{-1} \ln(L_N) \big\} \Big) \\  
                & \quad \ge \mathds P \Big( \tilde l_{\lfloor (1 - L_N^{-\varepsilon}) \nu L_N \rfloor - 2,>}^{(1),\omega} - \tilde l_{\lfloor (1 - L_N^{-\varepsilon}) \nu L_N \rfloor - 2,>}^{(2),\omega} > c_N \Big) \\
                & \qquad + \mathds P \Big( \max \Big( \big\{ \hat l^{j,\omega} : j \in J_{\lceil (1 + L_N^{-\varepsilon})L_N/2 \rceil} \big\} \backslash \big\{ \max \big( \{ \hat l^{j,\omega} : j \in J_{\lceil (1 + L_N^{-\varepsilon})L_N/2 \rceil} \} \big) \big\} \Big) \\
                & \qquad \qquad  \quad \in \big\{\hat l^{j,\omega} : j \in J_{\lfloor (1 - L_N^{-\varepsilon})L_N/2 \rfloor - 2} \big\} \Big\} \Big)  - 1\\
                & \qquad + \mathds P \Big( \max \Big( \big\{ \hat l^{j,\omega} : j \in J_{\lceil (1 + L_N^{-\varepsilon})L_N/2 \rceil} \big\} \backslash \big\{ \max \big( \{ \hat l^{j,\omega} : j \in J_{\lceil (1 + L_N^{-\varepsilon})L_N/2 \rceil} \} \big) \big\} \Big) \\
                & \qquad \qquad \qquad \in \big\{\hat l^{j,\omega} : j \in J_{\lceil (1+ L_N^{-\varepsilon}) L_N/2 \rceil} \backslash J_{\lfloor (1 - L_N^{-\varepsilon})L_N/2 \rfloor - 2} \big\} \Big) \\
                & \qquad + 2 \mathds P ( \Omega_N^{(\text{R})} \cap \Omega_N^{(\text{L})} \cap \Omega_N^{(0)} ) - 2\\
                & \qquad + 2 \mathds P \big( (1-L_N^{-\varepsilon})L_N/2 \le \kappa_N^{(L),\omega} \le (1+L_N^{-\varepsilon}) L_N/2 \big) - 2\\
                & \qquad + 2 \mathds P \big( (1-L_N^{-\varepsilon})L_N/2 \le \kappa_N^{(R),\omega} \le (1+L_N^{-\varepsilon}) L_N/2 \big) - 2\\
                & \qquad + 2 \mathds P \big( \max\{ \hat l^{j,\omega} : j \in J_{\lceil (1 + L_N^{-\varepsilon})L_N/2 \rceil} \backslash J_{\lfloor (1 - L_N^{-\varepsilon})L_N/2 \rfloor - 2} \} \le (1 - \varepsilon/2) \nu^{-1} \ln(L_N) \big) - 2\\
                & \qquad + 2 \mathds P \big( \max\{ \hat l^{j,\omega} : j \in J_{\lfloor (1 - L_N^{-\varepsilon})L_N/2 \rfloor - 2} \} \ge (1 - \varepsilon/3) \nu^{-1} \ln(L_N) \big\} \big) - 2 \ .
            \end{align*}
            Thus, by also using Lemma~\ref{Lemma 3.1} and Lemma~\ref{Lemma 3.2} as well as the fact that $\mathds P(A) + \mathds P(A^c) =1$ for any event $A \subset \Omega$,
            \begin{align*}
                \liminf\limits_{N \to \infty} \mathds P \Big( \llargest - \llargestzwei > c_N \Big) \ge \liminf\limits_{N \to \infty} \mathds P \Big( \tilde l_{\lfloor (1 - L_N^{-\varepsilon}) \nu L_N \rfloor - 2,>}^{(1),\omega} - \tilde l_{\lfloor (1 - L_N^{-\varepsilon}) \nu L_N \rfloor - 2,>}^{(2),\omega} > c_N \Big) \ .
            \end{align*}
        \end{proof}

        The next three lemmata are statements about the probability that the difference between the length of the largest interval and the length of the second largest interval of all subintervals $(I_N^{j,\omega})_{j \in \mathds Z}$ within $\Lambda_N$ is larger than a certain quantity. They are used in connection with results from Section~\ref{section energy bounds} to show a sufficient energy gap between the eigenvalues of the random Schrödinger operator $H_{N,\RPN}$ and to subsequently conclude the occurrence of type-I g-BEC in Section~\ref{secMainResults}.

        \begin{prop} \label{Proposition 3.4}
            For any $c>0$ and any $2 \le k \in \mathds N$,
                $$\mathds P(\llargestk{1} - \llargestk{2} > c) = \e^{- \nu c} \ .$$
            Furthermore, for any $c> 0$,
            \begin{align}
                \liminf\limits_{N \to \infty} \mathds P \left( \llargest - \llargestzwei > c \right) \ge \mathrm{e}^{-\nu c} \ .
            \end{align}
        \end{prop}
        \begin{proof}
    
            Let a $c>0$ and a $2 \le k \in \mathds N$ be arbitrarily given. We conclude, see also \cite[Section 6.3]{LenobleZagrebnovLuttingerSy},
            \begin{align}
                \begin{split} \label{zklxckjh2}
                    \mathds P(\llargestk{1} - \llargestk{2} > c) & = k(k-1) \int_{c}^{\infty} \int_0^{x-c} (1-\e^{-\nu y})^{k-2} \, \nu \e^{-\nu y}  \, \mathrm{d} y \, \nu \e^{-\nu x} \, \mathrm{d} x\\
                    & = k \int_{c}^{\infty} \left( 1 - \e^{-\nu(x - c)} \right)^{k-1} \nu \e^{-\nu x} \, \mathrm{d} x \\
                    & = \e^{- \nu c} \ .
                \end{split}
            \end{align}
            Moreover, with Lemma~\ref{Lemma 3.3},
            \begin{align*}
                \liminf\limits_{N \to \infty} \mathds P \left( \llargest - \llargestzwei > c \right) \ge \liminf\limits_{k \to \infty} \mathds P(\llargestk{1} - \llargestk{2} > c) = \e^{-\nu c} \ .
            \end{align*}
        \end{proof}

        \begin{prop} \label{Proposition 3.5}
            If $(c_N)_{N \in \mathds N} \subset [0,\infty)$ is a sequence with $\lim_{N \to \infty} c_N = 0$, then we have
            \begin{align}
                \lim\limits_{N \to \infty} \mathds P(\llargest - \llargestzwei > c_N) = 1 \ .
            \end{align}
        \end{prop}

        \begin{proof}
            With Lemma~\ref{Lemma 3.3} and the fact that
            \begin{align*}
                \mathds P(\llargestk{1} - \llargestk{2} > c_N) = \e^{- \nu c_N}
            \end{align*}
            for all $2 \le k \in \mathds N$ and all $N \in \mathds N$, see Proposition~\ref{Proposition 3.4}, we conclude, with an arbitrary $0 < \varepsilon < 1/2$,
            \begin{align*}
                & \liminf\limits_{N \to \infty} \mathds P \left( \llargest - \llargestzwei > c_N \right) \\
                & \quad \ge \liminf\limits_{N \to \infty} \mathds P \Big( \tilde l_{\lfloor (1 - L_N^{-\varepsilon}) \nu L_N \rfloor - 2,>}^{(1),\omega} - \tilde l_{\lfloor (1 - L_N^{-\varepsilon}) \nu L_N \rfloor - 2,>}^{(2),\omega} > c_N \Big)  \\
                & \quad = \liminf\limits_{N \to \infty} \e^{- \nu c_N} \\
                & \quad = 1 \ .
            \end{align*}
        \end{proof}    

        Lastly, we show that for any $c> 0$ one can obtain an arbitrarily high probability for the event $\{\llargest - \llargestN{j} > c\}$ by choosing $j \in \mathds N$ sufficiently large. In the proof of the next lemma, we use the limiting integrated density of the states of the Luttinger--Sy model \cite{luttinger1973bose, luttinger1973low, LenobleZagrebnovLuttingerSy}. For the convenience of the reader, we introduce this well-known quantum particle system and mention some important facts in the next remark.

        \begin{remark} \label{Remark Luttinger Sy model}
            The \emph{Luttinger--Sy model} is a system of noninteracting bosonic particles corresponding to the sequence of one-particle random Schrödinger operators $(H_{N,\mathrm{LS}})_{N \in \mathds N}$ that are defined, for $\mathds P$-almost all $\omega \in \Omega$ and all $N \in \mathds N$, as the strong resolvent limit $\gamma \to \infty$ of the self-adjoint operators defined by the quadratic form $q^{\omega}_{N,\gamma}$, see \eqref{quadratic form RSO gammaN}. Informally, it consists of noninteracting bosons on the real line in a Poisson random potential on $\mathds R$ with the single-site potential $\gamma \delta$ where $\delta$ is the Dirac delta function and $\gamma = \infty$ \cite[Section 2]{LenobleZagrebnovLuttingerSy}. $\mathds P$-almost surely, the one-particle Schrödinger operator $H_{N,\mathrm{LS}}^{\omega}$ has, for each $N \in \mathds N$, a purely discrete spectrum. We denote its eigenvalues, arranged in increasing order and each eigenvalue repeated according to its multiplicity, by $E_{N,\mathrm{LS}}^{j,\omega}$, $j \in \mathds N$.
            The sequence of the integrated densities of states $(\mathcal N_{N, \mathrm{LS}}^{\mathrm{I},\omega})_{N \in \mathds N}$ of the Luttinger--Sy model $\mathds P$-almost surely converges pointwise to the nonrandom limiting integrated density of states
            \begin{align}
                \mathcal N_{\infty, \mathrm{LS}}^{\mathrm{I}} : \mathds R \to [0,\infty), \ E \mapsto \mathcal N_{\infty, \mathrm{LS}}^{\mathrm{I}}(E) = \nu \dfrac{\e^{-\nu \pi E^{-1/2}}}{1 - \e^{-\nu \pi E^{-1/2}}} \mathds 1_{(0,\infty)}(E) \ .
            \end{align}
            In addition, the critical density $\rho_{c,\mathrm{LS}} = \int_{\mathds R} \mathcal B(E) \mathcal N_{\infty,\mathrm{LS}}(\mathrm{d} E)$ is finite, for any $\beta > 0$.
        \end{remark}

        \begin{prop} \label{Proposition 3.7}
            For all $\varepsilon > 0$ and all $\hat c > 0$ there is a $j \in \mathds N$ such that
            \begin{align}
                \mathds P(\llargest - \llargestN{j} > \hat c) \ge 1 - \varepsilon
            \end{align}
            for all but finitely many $N \in \mathds N$.
        \end{prop}
        \begin{proof}
            We follow the proof idea of \cite[Theorem 3.2]{KPS182} in large parts. Let $\hat c > 0$ and $0 < \varepsilon < 1$ be arbitrarily given. We choose an $\tilde c_1 >0$ such that
            \begin{align*}
                \mathds P \left( \llargest > \nu^{-1} \ln(L_N) - \tilde c_1 \right) \ge 1 - \varepsilon/2
            \end{align*}
            for all but finitely many $N \in \mathds N$; see, e.g., \cite[Theorem C.6]{KPS18} or \cite[Section 3.3, Lemma 3.4]{sznitman1998brownian} regarding a reason why such a choice of $\tilde c_1$ is possible.
    
            Next, we define $\tilde c_2 := \hat c + \tilde c_1 > 0$. For
            \begin{align*}
                \widehat E_{N} := \left( \dfrac{\nu \pi}{\ln(L_N) - \nu \tilde c_2} \right)^2
            \end{align*}
            we have
            \begin{align*}
                \E[\mathcal N_{N,\mathrm{LS}}^{\mathrm{I},\omega}(\widehat E_{N})] \le \mathcal N_{\infty, \mathrm{LS}}^{\mathrm{I}}(\widehat E_{N}) = \nu \dfrac{\e^{\nu \tilde c_2} L_N^{-1}}{1 - \e^{\nu \tilde c_2} L_N^{-1}}
            \end{align*}
            for all $N \in \mathds N$, see the information provided below \eqref{ENN le c Ninfty}. We multiply both sides of this inequality by $L_N$ and conclude
            \begin{align*}
                j \, \mathds P\left( \big| \big\{ i \in \mathds N : E_{N,\mathrm{LS}}^{i,\omega} \le \widehat E_{N} \big\} \big| \ge j \right) & \le \E \left[ \big| \big\{ i \in \mathds N : E_{N,\mathrm{LS}}^{i,\omega} \le \widehat E_{N} \big\} \big| \right] \le \dfrac{\nu \e^{\nu \tilde c_2}}{1 - \e^{\nu \tilde c_2}L_N^{-1}}
            \end{align*}
            for any $j \in \mathds N$ and all $N \in \mathds N$. Thus,
            \begin{align*}
                \mathds P \left( E_{N,\mathrm{LS}}^{j,\omega} > \widehat E_{N} \right) = \mathds P \left( \big| \big\{ i \in \mathds N : E_{N,\mathrm{LS}}^{i,\omega} \le \widehat E_{N} \big\} \big| \le j - 1 \right) \ge 1 - \dfrac{\nu \e^{\nu \tilde c_2}}{1 - \e^{\nu \tilde c_2} L_N^{-1}} \dfrac{1}{j}
            \end{align*}
            for any $j \in \mathds N$ and all $N \in \mathds N$. After choosing, for example, $j = \lceil 3 \nu \e^{\nu \tilde c_2} \varepsilon^{-1} \rceil$, that is, $j$ as the smallest natural number larger than or equal to $3 \nu \e^{\nu \tilde c_2} \varepsilon^{-1}$, we obtain for all sufficiently large $N \in \mathds N$
            \begin{align*}
                \mathds P \left( E_{N,\mathrm{LS}}^{j,\omega} \ge \widehat E_{N} \right) \ge 1 - \varepsilon/2 \ .
            \end{align*}
            Since $E_{N,\mathrm{LS}}^{j,\omega} \le ( \pi / \llargestN{j})^2$ $\mathds P$-almost surely, we have
            \begin{align*}
                \mathds P \left( \llargestN{j} \le \nu^{-1} \ln(L_N) - \tilde c_2 \right) & \ge 1 - \varepsilon/2
            \end{align*}
            for all sufficiently large $N \in \mathds N$.
    
            Finally, we conclude
            \begin{align*}
                & \mathds P \left( \llargestN{1} - \llargestN{j} > \hat c \right) \\
                & \quad \ge \mathds P \left( \llargest > \nu^{-1} \ln(L_N) - \tilde c_1, \llargestN{j} \le \nu^{-1} \ln(L_N) - \tilde c_2 \right) \\
                & \quad \ge \mathds P \left( \llargestN{j} \le \nu^{-1} \ln(L_N) - \tilde c_2 \right) + \mathds P \left( \llargest > \nu^{-1} \ln(L_N) - \tilde c_1 \right) - 1 \\
                & \quad \ge 1 - \varepsilon
            \end{align*}
            for all sufficiently large $N \in \mathds N$.
        \end{proof}

    \section{Energy bounds} \label{section energy bounds}

        The upper bound for the lowest eigenvalue in~\eqref{Def omega N j}, for any $0<\zeta_2<1$, is fulfilled with a probability that converges to one in the limit $N \to \infty$, see, e.g., Theorem 3.1 in Section 3.3 or Theorem 4.6 in Section 4.4 of \cite{sznitman1998brownian} (note that \cite{sznitman1998brownian} uses $-\frac{1}{2} \laplace$ instead of $-\laplace$ in \eqref{Definition RSO} and does not consider Poisson random potentials with a strength that converges to infinity). This can be shown by using the Rayleigh--Ritz variational principle and the fact that with high probability a sufficiently large interval free of atoms of $\PoissonrandommeasureRomega$ occurs within $\Lambda_N$, see Lemma~\ref{Lemma 3.2}. For the convenience of the reader, we provide more details in the next proposition. Recall from Section~\ref{SecPrelim} the definition of $\CSU$, $\CSUleft$, and $\CSUright$, in particular, \eqref{Def CSU}. Also, $(V_N)_{N \in \mathds N}$ is in this section either a Poisson random potential of fixed strength or with a strength that converges to infinity, see Remark~\ref{remark RP}.

        \begin{prop} \label{Proposition 4.1}
            We have
            \begin{align*}
                \lim\limits_{N \to \infty} \mathds P \Big( E_{N,\RPN}^{1,\omega} & \le \dfrac{\pi^2}{\big( \llargest - \CSU \big)^2} \Big) = 1 \ .
            \end{align*}
        \end{prop}

        \begin{proof}
            Firstly, note that
            \begin{align*}
                E_{N,\RPN}^{1,\omega} = \inf\Big\{ \int\limits_{\Lambda_N} |\varphi'(x) |^2 \ud x + \int\limits_{\Lambda_N} V_N(x) | \varphi(x) |^2 \ud x : \varphi \in \mathrm{H}_0^1(\Lambda_N), \|\varphi\|_{2} = 1 \Big\} \ .
            \end{align*}
            Let $\tilde x^{(1), \mathrm{left}, \omega}_N, \tilde x^{(1), \mathrm{right}, \omega}_N \in \mathds Z$ be the points that correspond to the largest interval $\Ilargest$ within $\Lambda_N$, that is, such that $\Ilargest = (\tilde x^{(1), \mathrm{left}, \omega}_N, \tilde x^{(1), \mathrm{right}, \omega}_N)$. We define $\widetilde I_{N,>}^{(1),\omega} := (\tilde x^{(1), \mathrm{left}, \omega}_N + \CSUright, \tilde x^{(1), \mathrm{right}, \omega}_N - \CSUleft)$, with the understanding that $(b,a) = \emptyset$ if $a \le b$, and choose
            \begin{align*}
                \varphi_N^{\omega} : \Lambda_N & \to \mathds R, \\
                x & \mapsto \varphi_N^{\omega}(x) := \Big( \frac{2}{\llargest - \CSU} \Big)^{1/2} \sin \Big( \frac{\pi (x - \tilde x^{(1), \mathrm{left}, \omega}_N - \CSUright)}{\llargest - \CSU}  \Big) \mathds 1_{\widetilde I_{N,>}^{(1),\omega}}(x)
            \end{align*}
            for $\mathds P$-almost all $\omega \in \Omega$ and all $N \in \mathds N$ with $\llargest > \CSU$. Thus,
            \begin{align*}
                & \int\limits_{\Lambda_N} V_N(x) | \varphi_N^{\omega}(x) |^2 \ud x  = 0
            \end{align*}
            and consequently
            \begin{align*}
                E_{N,\RPN}^{1,\omega} & \le \dfrac{\pi^2}{\big( \llargest - \CSU \big)^2}
            \end{align*}
            for $\mathds P$-almost all $\omega \in \Omega$ and all $N \in \mathds N$ such that $\llargest > \CSUleft + \CSUright$.
            After taking into account Lemma~\ref{Lemma 3.2}, we have proved this theorem.
        \end{proof}

        To prove that the gap between the lowest eigenvalue $E_{N,\RPN}^{1,\omega}$ and some higher eigenvalue $E_{N,\RPN}^{j,\omega}$, $2 \le j \in \mathds N$, in \eqref{Def omega N j} holds under certain circumstances with high probability as well, we next find a sufficiently good lower bound for $E_{N,\RPN}^{j,\omega}$, $2 \le j \in \mathds N$. As a first step, we use a lower bound for the first eigenvalue $E_{\left( -\laplace + V_N(\omega) \right)_{\IlargestN{j}}^{\text{Neu}}}^{1,\omega}$ of the operator $(-\laplace + V_N(\omega))_{\IlargestN{j}}^{\text{Neu}}$ with Neumann boundary conditions. Note that $(-\laplace + V_N(\omega))_{\IlargestN{j}}^{\text{Neu}}$ is the unique self-adjoint operator defined by the quadratic form
        \begin{equation}
            \begin{split}
                \mathrm{H}^1(\IlargestN{j}) \times \mathrm{H}^1(\IlargestN{j}) & \to \mathds C \\
                (\varphi,\psi) & \mapsto \int\limits_{\IlargestN{j}} \overline{\varphi^{\prime}(x)}\psi^{\prime}(x) \ud x +  \int\limits_{\IlargestN{j}} \overline{\varphi(x)} V_N(\omega,x) \psi(x) \ud x
            \end{split}
        \end{equation}
        on $L^2(\IlargestN{j})$ and has a purely discrete spectrum, see, e.g., \cite[Theorem 5.1 and p. 103]{pastur1992spectra} or \cite[p. 263]{reed1978analysisoperators}.

        \begin{lemma} \label{Lemma 4.2}
        Let $\omega \in \widetilde \Omega$, $N \in \mathds N$, and $j \in \mathds N$ be given. Furthermore, we assume that $\llargestN{j} \ge 2 \CSU$. Then for any $0 < a \le \CSUright$ and $0 < b \le \CSUleft$
        we have
        \begin{equation} \label{Inequality Lower bound}
            E_{(-\laplace + V_N(\omega))_{\IlargestN{j}}^{\mathrm{Neu}}}^{1,\omega} \ge \dfrac{\pi^2}{(\llargestN{j} - (\CSU-a-b))^2} - \dfrac{(4 \pi)^2 (8\pi+1)^2}{\widetilde \WWS_N^{a,b}} \dfrac{1}{(\llargestN{j} - \CSU)^2 \llargestN{j}}
        \end{equation}
        where
        \begin{equation}
            \widetilde \WWS_N^{a, b} := \WWS_N \min \Big\{\int_{\CSUright-a}^{\CSUright} u(x) \ud x , \int_{-\CSUleft}^{-\CSUleft+b} u(x) \ud x \Big\} > 0 \ .
        \end{equation}
    \end{lemma}
    \begin{proof}
        We follow \cite[Section 3.3, proof of Lemma 3.2]{sznitman1998brownian} in large parts. Let $0 < a \le \CSUright$ and $0 < b \le \CSUleft$ be arbitrarily given.
        Note that
        \begin{align*}
            & E_{(-\laplace + V_N(\omega))_{\IlargestN{j}}^{\mathrm{Neu}}}^{1,\omega} = \min\Big\{ \int\limits_{\IlargestN{j}} \left( |\varphi'|^2(x)  + V_N(\omega,x) |\varphi|^2(x) \right) \, \mathrm{d} x  : \varphi \in H^1(\IlargestN{j}), \|\varphi\|_{2}= 1 \Big\} \ .
        \end{align*}
        We conclude that there is a nonnegative function $\varphi_N \in \mathrm{H}^1(\IlargestN{j}) \cap \mathrm{C}(\overline{\IlargestN{j}})$ with $\|\varphi_N\|_2 = 1$ such that
        \begin{align*}
            & E_{(-\laplace + V_N(\omega))_{\IlargestN{j}}^{\mathrm{Neu}}}^{1,\omega} = \int\limits_{\IlargestN{j}} \left( |\varphi_N'|^2(x)  + V_N(\omega,x) |\varphi_N|^2(x) \right) \, \mathrm{d} x \ .
        \end{align*}
    
        In order to estimate $\int_{\IlargestN{j}} |\varphi_N'|^2(x) \, \mathrm{d} x$ from below, we use the Poincare inequality. We denote the two points that correspond to the $j$th largest interval by $\tilde x_N^{(j),\mathrm{left},\omega}$ and $\tilde x_N^{(j),\mathrm{right},\omega}$, $\IlargestN{j} = (\tilde x_N^{(j),\mathrm{left},\omega}, \tilde x_N^{(j),\mathrm{right},\omega})$. We subtract a linear function from $\varphi_N$ to introduce zeros in $[\tilde x_N^{(j),\mathrm{left},\omega} + \CSUright - a,  \tilde x_N^{(j),\mathrm{left},\omega} + \CSUright]$ and in $[\tilde x_N^{(j),\mathrm{right},\omega} - \CSUleft, \tilde x_N^{(j),\mathrm{right},\omega} - \CSUleft + b]$. 
        We indicate the location of the minimum of $\varphi_N$ in $[\tilde x_N^{(j),\mathrm{left},\omega} + \CSUright - a,  \tilde x_N^{(j),\mathrm{left},\omega} + \CSUright]$ and in $[\tilde x_N^{(j),\mathrm{right},\omega} - \CSUleft, \tilde x_N^{(j),\mathrm{right},\omega} - \CSUleft + b]$ by $\alpha_N^{\text{left}}$ and $\alpha_N^{\mathrm{right}}$, respectively. Furthermore, we denote the minimum of $\varphi_N$ in $[\tilde x_N^{(j),\mathrm{left},\omega} + \CSUright - a,  \tilde x_N^{(j),\mathrm{left},\omega} + \CSUright]$ and in $[\tilde x_N^{(j),\mathrm{right},\omega} - \CSUleft, \tilde x_N^{(j),\mathrm{right},\omega} - \CSUleft + b]$ by $\beta_N^{\text{left}}$ and $\beta_N^{\mathrm{right}}$, respectively. In conclusion,
        \begin{align*}
            \beta_N^{\text{left}} = \varphi_N(\alpha_N^{\text{left}}) = \min\{ \varphi_N(x) : x \in [\tilde x_N^{(j),\mathrm{left},\omega} + \CSUright - a,  \tilde x_N^{(j),\mathrm{left},\omega} + \CSUright] \}
            \shortintertext{and}
            \beta_N^{\mathrm{right}} = \varphi_N(\alpha_N^{\mathrm{right}}) = \min\{ \varphi_N(x) : [\tilde x_N^{(j),\mathrm{right},\omega} - \CSUleft, \tilde x_N^{(j),\mathrm{right},\omega} - \CSUleft + b] \} \ .
        \end{align*}
        After subtracting the linear function
        $$g_N: \overline{\IlargestN{j}} \to \mathds R, \quad x \mapsto g_N(x) := \dfrac{\beta_N^{\mathrm{right}}}{\alpha_N^{\mathrm{right}} - \alpha_N^{\text{left}}} (x - \alpha_N^{\text{left}}) + \dfrac{\beta_N^{\text{left}}}{\alpha_N^{\mathrm{right}} - \alpha_N^{\text{left}}} (\alpha_N^{\mathrm{right}} - x)$$
        from $\varphi_N$, we have $(\varphi_N - g_N)(\alpha_N^{\text{left}}) = (\varphi_N - g_N)(\alpha_N^{\mathrm{right}}) = 0$. Also,
        \begin{align*}
            \|\varphi_N'\|_2^2 & = \|\varphi_N' - g_N' + g_N'\|_2^2 \ge \left(\|(\varphi_N - g_N)' \|_2 - \|g_N'\|_2 \right)^2 \ .
        \end{align*}
        We extend $(\varphi_N - g_N)$ by symmetry to $[2 \tilde x_N^{(j),\mathrm{left},\omega} - \alpha_N^{\text{left}}, 2 \tilde x_N^{(j),\mathrm{right},\omega} - \alpha_N^{\mathrm{right}}]$,
        \begin{align*}
            & \hat \varphi_N - \hat g_N : [2 \tilde x_N^{(j),\mathrm{left},\omega}- \alpha_N^{\text{left}}, 2 \tilde x_N^{(j),\mathrm{right},\omega} - \alpha_N^{\mathrm{right}}] \to \mathds R \\
            & \qquad x \mapsto (\hat \varphi_N - \hat g_N)(x) :=
            \begin{cases}
                (\varphi_N - g_N)(x) & \text{if } x \in [\tilde x_N^{(j),\mathrm{left},\omega}, \tilde x_N^{(j),\mathrm{right},\omega}] \vspace{0.05in} \\
                (\varphi_N - g_N)(2 \tilde x_N^{(j),\mathrm{left},\omega}- x) & \text{if } x \in [2 \tilde x_N^{(j),\mathrm{left},\omega}- \alpha_N^{\text{left}}, \tilde x_N^{(j),\mathrm{left},\omega}) \vspace{0.05in} \\
                (\varphi_N - g_N)(2 \tilde x_N^{(j),\mathrm{right},\omega} - x) & \text{if } x \in (\tilde x_N^{(j),\mathrm{right},\omega}, 2 \tilde x_N^{(j),\mathrm{right},\omega} - \alpha_N^{\mathrm{right}}] \\
            \end{cases} \ .
        \end{align*}
        Note that $(\hat \varphi_N - \hat g_N)(2 \tilde x_N^{(j),\mathrm{left},\omega}- \alpha_N^{\text{left}}) = (\varphi_N - g_N)(\alpha_N^{\text{left}}) = 0$ and $(\hat \varphi_N - \hat g_N)(2 \tilde x_N^{(j),\mathrm{right},\omega} -\alpha_N^{\mathrm{right}}) = 0$.
        With the Poincare inequality, we now obtain
        \begin{align*}
            & \|(\varphi_N - g_N)'\|_2^2 \\
            & \quad = \int_{\alpha_N^{\text{left}}}^{\alpha_N^{\mathrm{right}}} (\varphi_N - g_N)'^2(x) \ud x + \dfrac{1}{2} \int_{2 \tilde x_N^{(j),\mathrm{left},\omega}- \alpha_N^{\text{left}}}^{\alpha_N^{\text{left}}} (\hat \varphi_N - \hat g_N)'^2(x) \ud x \\
            & \qquad + \, \dfrac{1}{2} \int_{\alpha_N^{\mathrm{right}}}^{2 \tilde x_N^{(j),\mathrm{right},\omega} - \alpha_N^{\mathrm{right}}} (\hat \varphi_N - \hat g_N)'^2(x) \ud x \\
            & \quad \ge \dfrac{\pi^2}{(\alpha_N^{\mathrm{right}} - \alpha_N^{\text{left}})^2} \int_{\alpha_N^{\text{left}}}^{\alpha_N^{\mathrm{right}}} (\varphi_N - g_N)^2 (x) \ud x \\
            & \qquad + \, \dfrac{\pi^2}{(2 \alpha_N^{\text{left}} - 2 \tilde x_N^{(j),\mathrm{left},\omega})^2} \int_{\tilde x_N^{(j),\mathrm{left},\omega}}^{\alpha_N^{\text{left}}} (\varphi_N - g_N)^2 (x) \ud x\\
            & \qquad + \, \dfrac{\pi^2}{(2 \tilde x_N^{(j),\mathrm{right},\omega} - 2 \alpha_N^{\mathrm{right}})^2} \int_{\alpha_N^{\mathrm{right}}}^{\tilde x_N^{(j),\mathrm{right},\omega}} (\varphi_N - g_N)^2(x) \ud x \\
            & \quad \ge \dfrac{\pi^2}{(\llargestN{j} - (\CSU - a - b))^2} \|\varphi_N - g_N\|_2^2 \ .
        \end{align*}
        Consequently, with $\beta_N := \beta_N^{\mathrm{right}} + \beta_N^{\text{left}}$,
        \begin{align*}
            \|\varphi_N'\|_2^2 & \ge \left( \dfrac{\pi \|\varphi_N -g_N\|_2}{\llargestN{j} - (\CSU - a - b)}  - \|g_N'\|_2 \right)^2 
            \ge \left( \dfrac{\pi (\|\varphi_N\|_2  - \| g_N\|_2)}{\llargestN{j} - (\CSU - a - b)}  - \beta_N \dfrac{1}{(\llargestN{j})^{1/2}}  \right)^2 \\
            & \ge \left( \dfrac{\pi}{\llargestN{j} - (\CSU - a - b)} - \left( 8\pi +1 \right) 2 \beta_N \dfrac{1}{(\llargestN{j})^{1/2}}  \right)^2 \ .
        \end{align*}
        In conclusion,
        \begin{align*}
            & E_{(-\laplace + V_N(\omega))_{\IlargestN{j}}^{\mathrm{Neu}}}^{1,\omega} \\
            & \quad \ge \, \left( \dfrac{\pi}{\llargestN{j} - (\CSU - a - b)} - \left( 8\pi +1 \right) 2 \beta_N \dfrac{1}{(\llargestN{j})^{1/2}}  \right)^2 \\
            & \qquad \qquad  + \, (\beta_N^{\mathrm{left}})^2 \int_{\tilde x_N^{(j),\mathrm{left},\omega} + \CSUright - a}^{\tilde x_N^{(j),\mathrm{left},\omega} + \CSUright} V_N(\omega,x) \ud x + (\beta_N^{\mathrm{right}})^2 \int_{\tilde x_N^{(j),\mathrm{right},\omega}-\CSUleft}^{\tilde x_N^{(j),\mathrm{right},\omega} - \CSUleft+b} V_N(\omega,x) \ud x \\
            & \quad \ge \, \left( \dfrac{\pi}{\llargestN{j} - (\CSU - a - b)} - \left( 8 \pi +1 \right) 2 \beta_N \dfrac{1}{(\llargestN{j})^{1/2}}  \right)^2 + \dfrac{1}{2} \widetilde \WWS_N^{a,b} \beta_N^2
        \end{align*}
        where $\widetilde \WWS_N^{a, b} = \WWS_N \min\big \{\int_{\CSUright-a}^{\CSUright} u(x) \ud x , \int_{-\CSUleft}^{-\CSUleft+b} u(x) \ud x \big\} > 0$.
    
        Next, we find a minimum of the right-hand side of the above inequality by computing its derivative with respect to $\beta_N$,
        \begin{align*}
            & - 2 \left( \dfrac{\pi}{\llargestN{j} - (\CSU - a - b)} - (8\pi+1) 2 \beta_N \dfrac{1}{(\llargestN{j})^{1/2}}  \right) (8\pi+1) 2 \dfrac{1}{(\llargestN{j})^{1/2}} + \widetilde \WWS_N^{a,b} \beta_N  
        \end{align*}
        which is zero if and only if
        \begin{align*}
            \beta_N & = \dfrac{4\pi(8\pi+1)}{\llargestN{j} - (\CSU - a - b)} \dfrac{1}{(\llargestN{j})^{1/2}} \Bigg[ \dfrac{8(8\pi+1)^2}{\llargestN{j}} + \widetilde \WWS_N^{a,b} \Bigg]^{-1} \ .
        \end{align*}
        Finally, we conclude
        \begin{align*}
            & E_{(-\laplace + V_N(\omega))_{\IlargestN{j}}^{\mathrm{Neu}}}^{1,\omega} \ge \dfrac{\pi^2}{(\llargestN{j} - (\CSU - a - b))^2} - \dfrac{(4 \pi)^2 (8\pi+1)^2}{\widetilde \WWS_N^{a,b}} \dfrac{1}{(\llargestN{j} - \CSU)^2 \llargestN{j}} \ .
        \end{align*}
    \end{proof}

    Lastly, in the next two propositions we state a lower bound for $E_{N,\RPN}^{2,\omega}$ and for $E_{N,\RPN}^{j,\omega}$, $2 \le j \in \mathds N$, respectively.

        \begin{prop} \label{Proposition 4.3}
            We have
            \begin{equation*}
                \lim\limits_{N \to \infty} \mathds P \Bigg( E_{N,\RPN}^{2,\omega} \ge \min\Big\{ E_{\left( -\laplace + V_N \right)_{\Ilargestzwei}^{\text{Neu}}}^{1,\omega}, \dfrac{9}{4} \dfrac{(\nu \pi)^2}{(\ln(L_N))^{2}} \Big\} \Bigg) = 1\ .
            \end{equation*}
        \end{prop}
        \begin{proof}
            Let $\omega \in \widetilde \Omega$ and $N \in \mathds N$ such that $\llargest \ge (3/4) \nu^{-1} \ln(L_N)$ be given. We replace the Dirichlet boundary conditions at $-L_N/2$ and $L_N/2$ with Neumann boundary conditions. Also, we put Neumann boundary conditions at all points $\hat x_j^{\omega} \in \Lambda_N$. As before, $I_N^{j,\omega} = (\hat x_{j}^{\omega}, \hat x_{j+1}^{\omega})$ for all $j \in \mathds Z$. In addition, we denote location of the two points that belongs to $\Ilargest$ by $\tilde x^{(1), \mathrm{left}, \omega}_N$ and $\tilde x^{(1), \mathrm{right}, \omega}_N$, that is, $\Ilargest = (\tilde x^{(1), \mathrm{left}, \omega}_N,\tilde x^{(1), \mathrm{right}, \omega}_N)$. We also put Neumann boundary conditions at $\tilde x^{(1), \mathrm{left}, \omega}_N + (4 \nu)^{-1} \ln(L_N)$ and $\tilde x^{(1), \mathrm{right}, \omega}_N - (4\nu)^{-1} \ln(L_N)$. We denote these intervals by
            \begin{align*}
                \tildeIlargestN{j} & :=
                \begin{cases}
                    (\tilde x^{(1), \mathrm{left}, \omega}_N + (4 \nu)^{-1} \ln(L_N), \tilde x^{(1), \mathrm{right}, \omega}_N - (4 \nu)^{-1} \ln(L_N)) \quad & \text{ if } j = 1 \\
                    \IlargestN{j} \quad & \text{ if } j \ge 2
                \end{cases} \ ,\\
                K_N^{1,\omega} & := (\tilde x^{(1), \mathrm{left}, \omega}_N, \tilde x^{(1), \mathrm{left}, \omega}_N + (4 \nu)^{-1} \ln(L_N)) \ ,
                \shortintertext{and}
                K_N^{2, \omega} & := (\tilde x^{(1), \mathrm{right}, \omega}_N - (4 \nu)^{-1} \ln(L_N), \tilde x^{(1), \mathrm{right}, \omega}_N) \ .
            \end{align*}
            Then we have
            \begin{align*}
                (-\laplace + V_N(\omega))_{\Lambda_N}^{\mathrm{D}} \ge \left(- \laplace + V_N \right)_{K_N^{1,\omega}}^{\text{Neu}} \oplus \left(- \laplace + V_N(\omega) \right)_{K_N^{2,\omega}}^{\text{Neu}} \oplus \bigoplus\limits_{j \in \mathds N : \, |\tildeIlargestN{j}| > 0} \left(- \laplace + V_N(\omega) \right)_{\tildeIlargestN{j}}^{\text{Neu}}  \ , 
            \end{align*}
            see also \cite[(5.11)]{jaeck2009nature} and \cite[Section XIII.15, Proposition 4]{reed1978analysisoperators}. 
            With Lemma~\ref{Lemma 3.2} we now conclude that the probability of the set of all $\omega \in \Omega$ such that
            \begin{align*}
                E_{(-\laplace + V_N(\omega))_{\tildeIlargestN{1}}^{\mathrm{Neu}}}^{2,\omega} & \ge (3/2)^2 (\nu \pi)^2  (\ln(L_N))^{-2} \ , \\
                E_{(-\laplace + V_N(\omega))_{K_N^{1,\omega}}^{\mathrm{Neu}}}^{1,\omega} & \ge 3 (\nu \pi)^2  (\ln(L_N))^{-2} \ , \\
                \shortintertext{and}
                E_{(-\laplace + V_N(\omega))_{K_N^{2,\omega}}^{\mathrm{Neu}}}^{1,\omega} & \ge 3 (\nu \pi)^2  (\ln(L_N))^{-2} \ .
            \end{align*}
            converges to one in the limit $N \to \infty$.
     
            The first inequality is due to the fact that $\lim_{N \to \infty} \mathds P(\llargest \le (4/3) \nu^{-1} \ln(L_N)) = 1$, see Lemma~\ref{Lemma 3.2}, and
            \begin{align*}
                E_{(-\laplace + V_N(\omega))_{\tildeIlargestN{j}}^{\mathrm{Neu}}}^{2,\omega} \ge E_{(-\laplace)_{\tildeIlargestN{j}}^{\mathrm{Neu}}}^{2,\omega} = \dfrac{\pi^2}{|\tildeIlargestN{j}|^2}
            \end{align*}
            for all $N \in \mathds N$ and $j \in \mathds N$ with $\tildeIlargestN{j} > 0$, see, e.g., \cite[p. 266]{reed1978analysisoperators}.
    
            We also show the last inequality in more details. The inequality in the middle can be shown almost identically. We proceed similarly as in the proof of Lemma~\ref{Lemma 4.2}. Let $N \in \mathds N$ and $\omega \in \widetilde \Omega$ such that $\llargest \ge (3/4) \nu^{-1} \ln(L_N)$ be given. We denote by $\varphi_N \in \mathrm{H}^1(K_N^{2,\omega}) \cap \mathrm{C}(\overline{K_N^{2,\omega}})$ the function with the properties $\varphi_N \ge 0$, $\|\varphi_N\|_{2} = 1$, and
            \begin{align*}
                & \inf \Big\{ \int\limits_{K_N^{2,\omega}} \big(|\varphi'|^2(x) + V_N(\omega,x)|\varphi|^2(x) \big) \ud x : \varphi \in \mathrm{H}^1(K_N^{1,\omega}), \|\varphi\|_{2} = 1 \Big\} \\
                & \qquad \ge \int\limits_{K_N^{2,\omega}} \big( |\varphi'_N|^2(x) + V_N(\omega,x) |\varphi_N|^2(x) \big) \ud x \\
                & \qquad \ge \int\limits_{K_N^{2,\omega}} |\varphi'_N|^2(x) \ud x + \int_{\tilde x^{(1), \mathrm{right}, \omega}_N - \CSUleft}^{\tilde x^{(1), \mathrm{right}, \omega}_N} V_N(\omega,x) |\varphi_N|^2(x) \ud x \ .
            \end{align*}
            We write $\beta_N$ for the minimum of $\varphi_N$ in $[\tilde x^{(1), \mathrm{right}, \omega}_N - \CSUleft,\tilde x^{(1), \mathrm{right}, \omega}_N]$.
            By $\tilde \varphi_N - \beta_N$ we denote the symmetric extension of $\varphi_N - \beta_N$ to $\widetilde K_N^{2,\omega} := [\tilde x^{(1), \mathrm{right}, \omega}_N - (2 \nu)^{-1} \ln(L_N), \tilde x^{(1), \mathrm{right}, \omega}_N]$. Using the Poincare inequality, we then conclude
            \begin{align*}
                \|\varphi_N'\|_{\mathrm{L}^2(K_N^{2,\omega})}^2 & = \|(\varphi_N - \beta_N)'\|_{\mathrm{L}^2(K_N^{2,\omega})}^2 = \dfrac{1}{2} \|(\tilde \varphi_N - \beta_N)'\|_{\mathrm{L}^2(\widetilde K_N^{2,\omega})}^2\\
                & \ge \dfrac{1}{2} \dfrac{\pi^2}{|\widetilde K_N^{2,\omega}|^2} \|\tilde \varphi_N - \beta_N\|_{{\mathrm{L}^2(\widetilde K_N^{2,\omega})}}^2 = \dfrac{\pi^2}{|\widetilde K_N^{2,\omega}|^2} \|\varphi_N - \beta_N\|_{\mathrm{L}^2(K_N^{2,\omega})}^2 \\
                & \ge \dfrac{\pi^2}{|\widetilde K_N^{2,\omega}|^2} \big( \|\varphi_N\|_{\mathrm{L}^2(K_N^{2,\omega})} - \|\beta_N\|_{\mathrm{L}^2(K_N^{2,\omega})} \big)^2 \ge \dfrac{\pi^2}{|\widetilde K_N^{2,\omega}|^2} \big( 1 - \beta_N |K_N^{2,\omega}|^{1/2} \big)^2 \\
                & \ge \dfrac{\pi^2}{|\widetilde K_N^{2,\omega}|^2} - 2 \dfrac{\pi^2 \beta_N}{|\widetilde K_N^{2,\omega}|^{3/2}} \ .
            \end{align*}
            Therefore,
            \begin{align*}
                E_{(-\laplace + V_N(\omega))_{K_N^{2,\omega}}^{\mathrm{Neu}}}^{1,\omega} & \ge \int\limits_{K_N^{2,\omega}} \big( |\varphi'_N|^2(x) + V_N(\omega,x) |\varphi_N|^2(x) \big) \ud x \\
                & \ge \dfrac{\pi^2}{|\widetilde K_N^{2,\omega}|^2} - 2 \dfrac{\pi^2 \beta_N}{|\widetilde K_N^{2,\omega}|^{3/2}} + \beta_N^2 \int_{\tilde x^{(1), \mathrm{right}, \omega}_N - \CSUleft}^{\tilde x^{(1), \mathrm{right}, \omega}_N} V_N(\omega,x) \ud x \\
                & \ge \dfrac{\pi^2}{|\widetilde K_N^{2,\omega}|^2} - \dfrac{2 \pi^4}{\int_{\tilde x^{(1), \mathrm{right}, \omega}_N - \CSUleft}^{\tilde x^{(1), \mathrm{right}, \omega}_N} V_N(\omega,x) \ud x} \dfrac{1}{|\widetilde K_N^{2,\omega}|^{3} } \\
                & \ge 3 (\nu \pi)^2 (\ln(L_N))^{-2}
            \end{align*}
            for all but finitely many $N \in \mathds N$.
        \end{proof}

        \begin{prop} \label{Proposition 4.4}
            For any $2 \le j \in \mathds N$,
            \begin{align*}
                \lim\limits_{N \to \infty} \mathds P \left( E_{N,\RPN}^{j,\omega} \ge \min\left\{ E_{\left(- \laplace + V_N(\omega) \right)_{\IlargestN{\lceil j/2 \rceil}}^{\text{Neu}}}^{1,\omega} , \dfrac{15}{9} \dfrac{(\nu \pi)^2}{(\ln (L_N))^{2}} \right\} \right) = 1 \ ,
            \end{align*}
            where $\lceil j/2 \rceil$ is the smallest natural number larger than or equal to $j/2$.
        \end{prop}
        \begin{proof}
            Let $N \in \mathds N$ and $\omega \in \widetilde \Omega$ be arbitrarily given. Similarly as in the proof of Proposition~\ref{Proposition 4.3}, we put Neumann boundary conditions at all points $x_j^{\omega}$ that are within $\Lambda_N$. Also, we replace the Dirichlet boundary conditions at $-L_N/2$ and $L_N/2$ with Neumann boundary conditions. We obtain
            \begin{align*}
                (-\laplace + V_N(\omega))_{\Lambda_N}^{\mathrm{D}} \ge \bigoplus\limits_{j \in \mathds Z : \, |I_N^{j,\omega}| > 0} \left(- \laplace + V_N(\omega) \right)_{I_N^{j,\omega}}^{\text{Neu}} \ .
            \end{align*}
            Moreover,
            \begin{align*}
                \left(- \laplace + V_N(\omega) \right)_{I_N^{j,\omega}}^{\text{Neu}} \ge \left(- \laplace \right)_{I_N^{j,\omega}}^{\text{Neu}}
            \end{align*}
            and
                $$E_{\left(- \laplace \right)_{I_N^{j,\omega}}^{\text{Neu}}}^{2,\omega} = \dfrac{\pi^2}{(l_N^{j,\omega})^2}$$
            for all $j \in \mathds N$ with $l_N^{j,\omega} = |I_N^{j,\omega}| > 0$.
            Also
            \begin{align*}
                E_{\left(- \laplace + V_N(\omega)\right)_{\IlargestN{j}}^{\text{Neu}}}^{3,\omega} \ge E_{\left(- \laplace \right)_{\IlargestN{j}}^{\text{Neu}}}^{3,\omega} \ge \dfrac{4 \pi^2}{(\llargestN{j})^2} \ ,
            \end{align*}
            see, e.g., \cite[p. 266]{reed1978analysisoperators}, and consequently
            \begin{align*}
                E_{\left(- \laplace + V_N(\omega)\right)_{\IlargestN{j}}^{\text{Neu}}}^{3,\omega} \ge \dfrac{4 \pi^2}{(\llargest)^2} \ge \dfrac{16}{9} \dfrac{(\nu \pi)^2}{(\ln (L_N))^{2}}
            \end{align*}
            for all $j \in \mathds N$ if $\llargest \le (3/2) \nu^{-1} \ln(L_N)$. Therefore with Lemma~\ref{Lemma 3.2},
            \begin{align*}
                \lim\limits_{N \to \infty} \mathds P \left( E_{N,\RPN}^{j,\omega} \ge \min\left\{ E_{\left(- \laplace + V_N(\omega) \right)_{\IlargestN{\lceil j/2 \rceil}}^{\text{Neu}}}^{1,\omega} , \dfrac{15}{9} \dfrac{(\nu \pi)^2}{(\ln (L_N))^{2}} \right\} \right) = 1
            \end{align*}
            for all $2 \le j \in \mathds N$.
        \end{proof}

    \section{Main results}\label{secMainResults}

        The next two propositions are crucial for the rest of this section and are an extension of \cite[Proposition 2.8]{KPSConditionTypeOneBEC2020} to Poisson random potentials whose strengths possibly converge to infinity. To prove them, we use Lemma~\ref{Lemma 5.1}. Subsequently, we show in Theorems~\ref{Theorem 5.4} and~\ref{Theorem 5.5} that the energy gap condition \eqref{Def omega N j} is fulfilled with high probability. Finally, we prove the occurrence of type-I g-BEC in Corollaries~\ref{Corollary 1} - \ref{Corollary 3}.
        If not otherwise stated, $(V_N)_{N \in \mathds N}$ is either a Poisson random potential of fixed strength or with a strength that converges to infinity, see Remark~\ref{remark RP}, in this section. 

        \begin{lemma} \label{Lemma 5.1}
            If $\rho \ge \rho_{c,V_1}$, then for any $\varepsilon > 0$ we $\mathds P$-almost surely have
            \begin{align*}
                \limsup\limits_{N \to \infty} \int\limits_{(0,\varepsilon]} \mathcal B(E - \mu_{N,\RPN}^{\omega}) \, \mathcal N_{N,\RPN}^{\omega}(\mathrm{d} E) \le \rho - \int\limits\limits_{(\varepsilon,\infty)} \mathcal B(E) \, \mathcal N_{\infty, \RP_1}(\mathrm{d} E) + \frac{6}{\beta \varepsilon} \mathcal N_{\infty,V_1}^{\mathrm{I}}(4 \varepsilon)
            \end{align*}
            as well as
            \begin{align*}
                &\limsup\limits_{N \to \infty} \int\limits_{(0,E_{N,V_N}^{j-1,\omega}]} \mathcal B(E - \mu_{N,\RPN}^{\omega}) \, \mathcal N_{N,\RPN}^{\omega}(\mathrm{d} E) \le \rho_{0,V_1}
            \end{align*}
            for all $2 \le j \in \mathds N$.
        \end{lemma}

        \begin{proof}
            Since for $\mathds P$-almost all $\omega \in \Omega$, $\lim_{N \to \infty} \llargest =\infty$, see for example \cite[Chapter 3, Lemma 3.4]{sznitman1998brownian} or an appropriate version of Lemma~\ref{Lemma 3.2} (in combination with the Borel-Cantelli lemma), we have $\lim\limits_{N \to \infty} E_{N,\RPN}^{j-1,\omega}=0$ and therefore
            \begin{align*}
                \limsup\limits_{N \to \infty} \int\limits_{(0,E_{N,V_N}^{j-1,\omega}]} \mathcal B(E - \mu_{N,\RPN}^{\omega}) \, \mathcal N_{N,\RPN}^{\omega}(\mathrm{d} E) \le \limsup\limits_{N \to \infty} \int\limits_{(0,\varepsilon]} \mathcal B(E - \mu_{N,\RPN}^{\omega}) \, \mathcal N_{N,\RPN}^{\omega}(\mathrm{d} E)
            \end{align*}
            $\mathds P$-almost surely for any $\varepsilon > 0$ and for all $j \ge 2$. Furthermore, using integration by parts (for Lebesgue--Stieltjes integrals, see, e.g., \cite[Theorem 21.67]{hewitt1965real}), the fact that $\mu_{N,V_N}^{\omega} < E_{N,V_N}^{1,\omega}$ for $\mathds P$-almost all $\omega \in \Omega$ and all $N \in \mathds N$, and the Lifshitz tails property \eqref{Lifshitz tail}, we conclude for any $\varepsilon > 0$
            \begin{align*}
                & \limsup\limits_{N \to \infty} \int\limits_{(0,\varepsilon]} \mathcal B(E - \mu_{N,\RPN}^{\omega}) \, \mathcal N_{N,\RPN}^{\omega}(\mathrm{d} E)\\
                & \quad \le \, \limsup\limits_{N \to \infty} \lim\limits_{m \to 0} \Big[ \int\limits_{[m,\varepsilon]} \mathcal N_{N,\RPN}^{\mathrm{I}, \omega}(E-) (- \mathcal B'(E - \mu_{N,\RPN}^{\omega})) \ud E + \mathcal N_{N,\RPN}^{\mathrm{I}, \omega}(\varepsilon+) \mathcal B(\varepsilon - \mu_{N,\RPN}^{\omega}) \Big]\\
                & \quad \le \, \limsup\limits_{N \to \infty} \lim\limits_{m \to 0} \Big[ \int\limits_{[m,\varepsilon]} \mathcal N_{N,\RP_1}^{\mathrm{I}, \omega}(E-) (- \mathcal B'(E - \mu_{N,\RPN}^{\omega})) \ud E + \mathcal N_{N,\RP_1}^{\mathrm{I}, \omega}(2\varepsilon) \mathcal B(\varepsilon/2) \Big]\\
                & \quad \le \, \limsup\limits_{N \to \infty} \lim\limits_{m \to 0} \Big[ \int\limits_{[m,\varepsilon]} \mathcal B(E - \mu_{N,\RPN}^{\omega}) \, \mathcal N_{N,\RP_1}^{\omega}(\mathrm{d} E) + \mathcal N_{N,\RP_1}^{\mathrm{I}, \omega}(m-) \mathcal B(m - \mu_{N,\RPN}^{\omega}) \\
                & \qquad \qquad \qquad \qquad \qquad + \, \mathcal N_{N,\RP_1}^{\mathrm{I}, \omega}(2\varepsilon) \mathcal B(\varepsilon/2) \Big]\\
                & \quad \le \, \limsup\limits_{N \to \infty} \int\limits_{(0,\varepsilon]} \mathcal B(E - \mu_{N,\RPN}^{\omega}) \, \mathcal N_{N,\RP_1}^{\omega}(\mathrm{d} E) + \mathcal B(\varepsilon/2) \limsup\limits_{N \to \infty} \mathcal N_{N,V_1}^{\mathrm{I}, \omega}(2\varepsilon)
            \end{align*}
            $\mathds P$-almost surely, where $f(a-)$ and $f(a+)$ indicate the left and right limit of $f$ at $a$, respectively.
 
            Regarding the last term, we have for any $\widetilde \varepsilon > 0$
            \begin{align}
                \begin{split}\label{upper bound NNI}
                    & \limsup\limits_{N \to \infty} \mathcal N_{N,V_1}^{\mathrm{I}, \omega}(\widetilde \varepsilon) = \limsup\limits_{N \to \infty} \int\limits_{(0,\widetilde \varepsilon]} \, \mathcal N_{N,\RP_1}^{\omega}(\mathrm{d} E) \\
                    & \quad \le \, \limsup\limits_{N \to \infty} \int\limits_{\mathds R} \big[ (x+1) \mathds 1_{[-1,0]}(x) + \mathds 1_{[0,\widetilde \varepsilon]}(x) - (1 - \frac{x-\widetilde \varepsilon}{\widetilde \varepsilon}) \mathds 1_{[\widetilde \varepsilon,2\widetilde \varepsilon]}(x) \big] \, \mathcal N_{N,\RP_1}^{\omega}(\mathrm{d} E) \\
                    & \quad \le \mathcal N_{\infty,V_1}^{\mathrm{I}}(2\widetilde \varepsilon)
                \end{split}
            \end{align}
            $\mathds P$-almost surely due to the $\mathds P$-almost sure convergence of $(\mathcal N_{N,V_1}^{\omega})_{N \in \mathds N}$ to $\mathcal N_{\infty,V_1}$ in the vague sense.
            Also, for any $N \in \mathds N$ and any $\omega \in \widetilde \Omega$ we have $\mathcal N_{N,V_1}^{\mathrm{I},\omega}(m) = 0$ for all sufficiently small $m > 0$ since $E_{N,V_1}^{1,\omega} > 0$, $\mathcal B(m - \mu_N^{\omega}) \le \mathcal B(m)$ for all $m > 0$ if $\mu_{N,V_N}^{\omega} \le 0$, and $\mathcal B(m - \mu_N^{\omega}) = 0$ for all sufficiently small $m >0$ if $\mu_{N,V_N}^{\omega} > 0$.
 
            After using Lemmata~\ref{Lemma A.3},~\ref{Lemma A.4}, and~\ref{Lemma A.6} and the Lifshitz tails property \eqref{Lifshitz tail} we have proved this lemma. In particular,
            \begin{align*}
                & \limsup\limits_{N \to \infty} \int\limits_{(0,E_{N,V_N}^{j-1,\omega}]} \mathcal B(E - \mu_{N,\RPN}^{\omega}) \, \mathcal N_{N,\RPN}^{\omega}(\mathrm{d} E) \\
                & \quad \le \, \lim\limits_{\varepsilon \searrow 0} \limsup\limits_{N \to \infty} \int\limits_{(0,\varepsilon]} \mathcal B(E - \mu_{N,\RPN}^{\omega}) \, \mathcal N_{N,\RPN}^{\omega}(\mathrm{d} E) \\
                & \quad \le \, \lim\limits_{\varepsilon \searrow 0}  \limsup\limits_{N \to \infty} \int\limits_{(0, \varepsilon]} \mathcal B(E - \mu_{N,\RPN}^{\omega}) \, \mathcal N_{N,\RP_1}^{\omega}(\mathrm{d} E) + \lim\limits_{\varepsilon \searrow 0} \mathcal B(\varepsilon/2) \limsup\limits_{N \to \infty} \mathcal N_{N,V_1}^{\mathrm{I}, \omega}(2\varepsilon) \\
                & \quad \le \, \rho - \lim\limits_{\varepsilon \searrow 0}  \liminf\limits_{N \to \infty} \int\limits_{(\varepsilon, \infty)} \mathcal B(E - \mu_{N,\RPN}^{\omega}) \, \mathcal N_{N,\RP_1}^{\omega}(\mathrm{d} E) + \lim\limits_{\varepsilon \searrow 0} \mathcal B(\varepsilon/2) \mathcal N_{\infty,V_1}^{\mathrm{I}}(4\varepsilon) \\
                & \quad = \rho - \rho_{c,\RP_1}
            \end{align*}
            $\mathds P$-almost surely.
        \end{proof}

        \begin{prop}\label{Proposition 5.2}
            If $\rho \ge \rho_{c,V_1}$, then for any $2 \le j \in \mathds N$
            \begin{align}
                \liminf_{N \rightarrow \infty} \mathds{E}\, \int_{(0,E_{N,\RPN}^{j-1,\omega}]}\mathcal{B}(E-\mu_{N,\RPN}^{\omega}) \, \mathcal{N}_{N,\RPN}^{\omega}(\mathrm{d} E) \geq \rho_{0,V_1} \liminf_{N\rightarrow \infty} \mathds{P}(\Omega_{N,\RPN}^{j,\zeta_1,\zeta_2})\ .
            \end{align}
        \end{prop}

        \begin{proof}
            We follow the proof of \cite[Proposition 2.8]{KPSConditionTypeOneBEC2020} in large parts and extend it to allow for Poisson random potentials whose strengths converge to infinity. Since for all $\omega \in \widetilde \Omega$ and all $N \in \mathds N$ we have $\rho = \int_{(0,\infty)} \mathcal{B}(E - \mu_{N,\RPN}^{\omega})  \, \mathcal N_{N,\RPN}^{\omega}(\mathrm{d} E)$, we conclude
            \begin{align*}
                & \E \int\limits_{(0,E_{N,\RPN}^{j-1,\omega}]} \mathcal{B}(E - \mu_{N,\RPN}^{\omega})  \, \mathcal N_{N,\RPN}^{\omega}(\mathrm{d} E) \\
                & \quad = \rho\, \mathds P(E_{N,\RPN}^{j-1,\omega} < \varepsilon) - \E \, \mathds 1_{\{E_{N,\RPN}^{j-1,\omega} < \varepsilon\}}(\omega) \int\limits_{(E_{N,\RPN}^{j-1,\omega},\varepsilon]} \mathcal{B}(E - \mu_{N,\RPN}^{\omega})  \, \mathcal N_{N,\RPN}^{\omega}(\mathrm{d} E) \\
                & \qquad \qquad \qquad \qquad \quad \,  - \E \, \mathds 1_{\{E_{N,\RPN}^{j-1,\omega} < \varepsilon\}}(\omega) \int\limits\limits_{(\varepsilon,\infty)} \mathcal{B}(E - \mu_{N,\RPN}^{\omega})  \, \mathcal N_{N,\RPN}^{\omega}(\mathrm{d} E) \\
                & \qquad \qquad \qquad \qquad \quad \,  + \E \, \mathds 1_{\{E_{N,\RPN}^{j-1,\omega} \ge \varepsilon\}}(\omega) \int\limits_{(0,E_{N,\RPN}^{j-1,\omega}]} \mathcal{B}(E - \mu_{N,\RPN}^{\omega})  \, \mathcal N_{N,\RPN}^{\omega}(\mathrm{d} E) \label{proof BEC type ns LSM 4}
            \end{align*}
            for all $\varepsilon > 0$ and all $N \in \mathds N$. Recall that since $\lim_{N \to \infty} \llargest =\infty$ $\mathds P$-almost surely we have $\mathds P(\lim_{N \to \infty} E_{N,\RPN}^{j-1,\omega}=0)=1$. Thus, the last term converges to zero in the limit $N \to \infty$.
	
            Firstly, using reverse Fatou lemma and integration by parts twice we obtain
            \begin{align*}
                & \lim\limits_{\varepsilon \searrow 0} \limsup\limits_{N \to \infty} \mathds E \int\limits\limits_{(\varepsilon,\infty)} \mathcal{B}(E - \mu_{N,\RPN}^{\omega}) \, \, \mathcal N_{N,\RPN}^{\omega}(\mathrm{d} E) \\
                & \quad \le \lim\limits_{\varepsilon \searrow 0} \mathds E \limsup\limits_{N \to \infty} \int\limits\limits_{(\varepsilon,\infty)} \mathcal{B}(E - \mu_{N,\RPN}^{\omega}) \, \, \mathcal N_{N,\RPN}^{\omega}(\mathrm{d} E) \\
                & \quad \le \lim\limits_{\varepsilon \searrow 0} \mathds E \Big[ \limsup\limits_{N \to \infty} \lim\limits_{M \to \infty} \int\limits_{[\varepsilon,M]} \mathcal N_{N,\RPN}^{\mathrm{I},\omega}(E-) \, ( - \mathcal{B'}(E - \mu_{N,\RPN}^{\omega})) \ud E  \\
                & \qquad \qquad + \, \limsup\limits_{N \to \infty} \lim\limits_{M \to \infty} \mathcal N_{N,\RPN}^{\mathrm{I},\omega}(M+) \, \mathcal{B}(M - \mu_{N,\RPN}^{\omega})\Big] \\
                & \quad \le \lim\limits_{\varepsilon \searrow 0} \mathds E \Big[ \limsup\limits_{N \to \infty} \lim\limits_{M \to \infty} \int\limits_{[\varepsilon,M]} \mathcal N_{N,\RP_1}^{\mathrm{I},\omega}(E-) \, ( - \mathcal{B'}(E - \mu_{N,\RPN}^{\omega})) \ud E  \\
                & \qquad \qquad + \, \limsup\limits_{N \to \infty} \lim\limits_{M \to \infty} \mathcal N_{N,\RP_1}^{\mathrm{I},\omega}(2 M) \, \mathcal{B}(M - \mu_{N,\RPN}^{\omega})\Big] \\
                & \quad \le \lim\limits_{\varepsilon \searrow 0} \E \Big[ \limsup\limits_{N \to \infty} \lim\limits_{M \to \infty} \int\limits_{[\varepsilon,M]} \mathcal{B}(E - \mu_{N,\RPN}^{\omega}) \, \, \mathcal N_{N,\RP_1}^{\omega}(\mathrm{d} E) + \limsup\limits_{N \to \infty} \mathcal N_{N,\RP_1}^{\mathrm{I},\omega}(\varepsilon-) \, \mathcal{B}(\varepsilon - \mu_{N,\RPN}^{\omega})\Big] \\
                & \quad \le \lim\limits_{\varepsilon \searrow 0} \E \Big[ \limsup\limits_{N \to \infty} \int\limits\limits_{(\varepsilon,\infty)} \mathcal{B}(E - \mu_{N,\RPN}^{\omega}) \, \, \mathcal N_{N,\RP_1}^{\omega}(\mathrm{d} E) + 2 \mathcal{B}(\varepsilon/2) \limsup\limits_{N \to \infty} \mathcal N_{N,\RP_1}^{\mathrm{I},\omega}(2 \varepsilon) \Big]\\
                & \quad \le \rho_{c,\RP_1}
            \end{align*}
            due to Lemma~\ref{Lemma A.3}, Lemma~\ref{Lemma A.6}, definition~\eqref{Definition critical density}, inequality \eqref{upper bound NNI}, and the Lifshitz tails property \eqref{Lifshitz tail}.
            In addition, note that $\mu_{N,\RPN}^{\omega} < E_{N,V_N}^{1,\omega} \le (\pi/L_N)^2$ for $\mathds P$-almost all $\omega \in \Omega$ and for all $N \in \mathds N$ as well as that $\mathcal N_{N,\RP_1}^{\mathrm{I},\omega}(E) \le \pi^{-1} E^{1/2}$ for $\mathds P$-almost all $\omega \in \Omega$, for all $E \ge 0$, and all $N \in \mathds N$ (see, e.g., \cite[Remark 3.2.2]{pechmanndiss} for more details) and therefore $\lim_{M \to \infty} \mathcal N_{N,\RP_1}^{\mathrm{I},\omega}(2 M) \, \mathcal{B}(M - \mu_{N,\RPN}^{\omega}) = 0$ $\mathds P$-almost surely for all $N \in \mathds N$.
	
            Secondly, since $\mathcal{B}(E - \mu_{N,\RPN}^{\omega}) \le ( \beta (E - E_{N,\RPN}^{1,\omega}) )^{-1} \le \beta^{-1} N^{1 - \zeta_1}$
            for all $\omega \in \Omega_{N,\RPN}^{j,\zeta_1,\zeta_2}$ and $E \ge E_{N,\RPN}^{j,\omega}$, since $\mathcal N_{N,\RPN}^{\mathrm{I},\omega}(E) \ge 0$ for $\mathds P$-almost all $\omega \in \Omega$, for all $N \in \mathds N$, and for all $E \in \mathds R$, with inequality~\eqref{ENN le c Ninfty} and due to the Lifshitz tails property \eqref{Lifshitz tail}, there is a constant $C_1 > 0$ such that
            \begin{align*}
                & \lim\limits_{N \to \infty} \int\limits_{\Omega_{N,\RPN}^{j,\zeta_1,\zeta_2}} \mathds 1_{\{E_{N,\RPN}^{j-1,\omega} < (1 + (\zeta_1 + \zeta_2)/2) \nu \pi / \ln(L_N)]^2\}}(\omega) \cdot \\
                & \qquad \qquad \qquad \Big[ \int\limits_{(E_{N,\RPN}^{j-1,\omega},[(1 + (\zeta_1 + \zeta_2)/2) \nu \pi/ \ln (L_N)]^2]} \mathcal{B}(E - \mu_{N,\RPN}^{\omega}) \, \, \mathcal N_{N,\RPN}^{\omega}(\mathrm{d} E) \Big] \, \mathds P(\mathrm{d} \omega) \\
                & \quad \le \lim\limits_{N \to \infty} \beta^{-1} N^{1-\zeta_1} \E \Big[ \mathcal N_{N,\RPN}^{\mathrm{I},\omega} \big([(1 + (\zeta_1 + \zeta_2)/2) \nu \pi/ \ln (L_N)]^2 \big) \Big] \\
                & \quad \le \lim\limits_{N \to \infty} \beta^{-1} N^{1-\zeta_1} \E \Big[ \mathcal N_{N,\RP_1}^{\mathrm{I},\omega} \big( [(1 + (\zeta_1 + \zeta_2)/2) \nu \pi/ \ln (L_N)]^2 \big) \Big] \\
                & \quad \le \constantENNleNinfty \lim\limits_{N \to \infty} \beta^{-1} N^{1-\zeta_1} \mathcal N_{\infty, \RP_1}^{\mathrm{I}} \big( [(1 + (\zeta_1 + \zeta_2)/2) \nu \pi/ \ln (L_N)]^2 \big)\\
                & \quad = 0 \ .
            \end{align*}
            In addition, we have $E \ge c_2 E_{N,\RPN}^{1,\omega}$ for all $E \ge [(1 + (\zeta_1 + \zeta_2)/2) \nu \pi / \ln(L_N)]^2$ and all $\omega \in \Omega_{N,\RPN}^{j,\zeta_1,\zeta_2}$ where
            $c_2 := ([1 + (\zeta_1 + \zeta_2)/2)]/[1 + \zeta_2])^2 > 1$. Therefore, using that $E - \mu_{N,\RPN}^{\omega} \ge E - E_{N,\RPN}^{1,\omega} \ge (1 - c_2^{-1}) E$ for all but finitely many $N \in \mathds N$ in this case, integration by parts, and the Fubini--Tonelli theorem while proceeding similarly as above, one obtains
            \begin{align*}
                & \lim\limits_{\varepsilon \searrow 0} \limsup\limits_{N \to \infty} \int\limits_{\Omega_{N,\RPN}^{j,\zeta_1,\zeta_2}} \Big[ \int\limits_{([(1 + (\zeta_1 + \zeta_2)/2) \nu \pi / \ln(L_N)]^2,\varepsilon]}\mathcal{B}(E - \mu_{N,\RPN}^{\omega}) \, \, \mathcal N_{N,\RPN}^{\omega}(\mathrm{d} E) \Big] \, \mathds P(\mathrm{d} \omega) \\
                & \quad \le \dfrac{1}{\beta(1 - c_2^{-1})}  \lim\limits_{\varepsilon \searrow 0} \limsup\limits_{N \to \infty} \int\limits_{\Omega_{N,\RPN}^{j,\zeta_1,\zeta_2}} \Big[ \int\limits_{([(1 + (\zeta_1 + \zeta_2)/2) \nu \pi / \ln(L_N)]^2,\varepsilon]} E^{-1}  \, \mathcal N_{N,\RPN}^{\omega} (\mathrm{d} E) \Big] \, \mathds P(\mathrm{d} \omega) \\
                & \quad \le \dfrac{1}{\beta(1 - c_2^{-1})}  \lim\limits_{\varepsilon \searrow 0} \limsup\limits_{N \to \infty} \Big[ \varepsilon^{-1} \E \big[ \mathcal N_{N,\RPN}^{\mathrm{I},\omega} ( 2\varepsilon) \big] \\
                & \qquad \qquad \qquad \qquad \qquad \qquad \qquad + \int\limits_{[(1 + (\zeta_1 + \zeta_2)/2) \nu \pi / \ln(L_N)]^2}^{\varepsilon} \E \big[ \mathcal N_{N,\RPN}^{\mathrm{I},\omega} (E) \big]  E^{-2} \, \mathrm{d} E \Big] \\
                & \quad \le \dfrac{1}{\beta(1 - c_2^{-1})}  \lim\limits_{\varepsilon \searrow 0} \limsup\limits_{N \to \infty} \Big[ \varepsilon^{-1} \E \big[ \mathcal N_{N,\RP_1}^{\mathrm{I},\omega} ( 2\varepsilon) \big] \\
                & \qquad \qquad \qquad \qquad \qquad \qquad \qquad + \int\limits_{[(1 + (\zeta_1 + \zeta_2)/2) \nu \pi / \ln(L_N)]^2}^{\varepsilon} \E \big[ \mathcal N_{N,\RP_1}^{\mathrm{I},\omega} (E) \big]  E^{-2} \, \mathrm{d} E \Big] \\
                & \quad = 0 \ .
            \end{align*}

            Lastly,
            \begin{align*}
                & \lim\limits_{\varepsilon \searrow 0} \limsup\limits_{N \to \infty} \int\limits_{\Omega \backslash \Omega_{N,\RPN}^{j,\zeta_1,\zeta_2}} \mathds 1_{\{E_{N,\RPN}^{j-1,\omega} < \varepsilon\}}(\omega) \Big[ \int\limits_{(E_{N, \RPN}^{j-1,\omega},\varepsilon]} \mathcal{B}(E - \mu_{N,\RPN}^{\omega}) \, \mathcal N_{N,\RPN}^{\omega}(\mathrm{d} E) \Big] \, \mathds P(\mathrm{d} \omega) \\
                & \quad \le \lim\limits_{\varepsilon \searrow 0} \limsup\limits_{N \to \infty} \int\limits_{\Omega \backslash \Omega_{N,\RPN}^{j,\zeta_1,\zeta_2}} \Big[ \int\limits_{(0,\varepsilon]} \mathcal{B}(E - \mu_{N,\RPN}^{\omega}) \, \mathcal N_{N,\RPN}^{\omega}(\mathrm{d} E) \Big] \, \mathds P(\mathrm{d} \omega) \\
                & \quad \le \rho_{0,V_1} \limsup\limits_{N \to \infty} \mathds P(\Omega \backslash \Omega_{N,\RPN}^{j,\zeta_1,\zeta_2}) \ .
            \end{align*}
            We show the last step in more details. Let a sequence $(\varepsilon_i)_{i \in \mathds N} \subset (0,\infty)$ with $\lim_{i \to \infty} \varepsilon_i = 0$ be arbitrarily given.
            With Lemma~\ref{Lemma 5.1} we conclude the following: For all $\eta > 0$ there exists an $\widetilde I \in \mathds N$ such that for all $i \ge \widetilde I$ we $\mathds P$-almost surely have 
                $$\limsup\limits_{N \to \infty} \int\limits_{(0,\varepsilon_i]} \mathcal{B}(E - \mu_{N, \RP_{N}}^{\omega}) \, \mathcal N_{N,\RPN}^{\omega}(\mathrm{d} E) \le \rho_{0,V_1} + \eta \ .$$
            Consequently, for all $\eta > 0$ there exists an $\widetilde I \in \mathds N$ such that for all $i \ge \widetilde I$ the sequence of random variables
                $$\big( \max\big\{ 0, \int\limits_{(0,\varepsilon_i]} \mathcal{B}(E - \mu_{N, \RP_{N}}^{\omega}) \, \mathcal N_{N,\RPN}^{\omega}(\mathrm{d} E) - \rho_{0,V_1} - \eta \big\} \big)_{N \in \mathds N}$$
            converges $\mathds P$-almost surely to zero in the limit $N \to \infty$.
            Therefore, for all $\eta >0$ and for all but finitely many $i \in \mathds N$, $\lim_{N \to \infty} \mathds P(\widehat \Omega_N^{i, \eta}) = 1$ where
            \begin{align*}
                \widehat \Omega_N^{i,\eta} := \Big\{\omega \in \widetilde \Omega: \int\limits_{(0,\varepsilon_i]} \mathcal{B}(E - \mu_{N, \RPN}^{\omega}) \, \mathcal N_{N,\RPN}^{\omega}(\mathrm{d} E) \le \rho_{0,V_1} + \eta \Big\} \ .
            \end{align*}
            We conclude that for any $\eta > 0$,
            \begin{align*}
                &\lim\limits_{i \to \infty} \limsup\limits_{N \to \infty} \int\limits_{\Omega \backslash \Omega_{N,\RP_{N}}^{j,\zeta_1,\zeta_2}} \Big[ \int\limits_{(0,\varepsilon_i]} \mathcal{B}(E - \mu_{N, \RP_{N}}^{\omega}) \, \mathcal N_{N,\RPN}^{\omega}(\mathrm{d} E) \Big] \, \mathds P(\mathrm{d} \omega) \\
                & \qquad \le \lim\limits_{i \to \infty} \limsup\limits_{N \to \infty} \int\limits_{ \widehat \Omega_N^{i,\eta} \backslash \Omega_{N,\RP_{N}}^{j,\zeta_1,\zeta_2}} \Big[ \int\limits_{(0,\varepsilon_i]} \mathcal{B}(E - \mu_{N, \RP_{N}}^{\omega}) \, \mathcal N_{N,\RPN}^{\omega}(\mathrm{d} E) \Big] \, \mathds P(\mathrm{d} \omega) \\
                & \qquad \qquad + \, \rho \lim\limits_{i \to \infty} \limsup\limits_{N \to \infty} \mathds P(\Omega \backslash \widehat \Omega_N^{i,\eta}) \\
                & \qquad \le \Big[ \rho_{0,V_1} + \eta \Big] \lim\limits_{i \to \infty}  \limsup\limits_{N \to \infty} \mathds P(\widehat \Omega_N^{i,\eta} \backslash \Omega_{N,\RP_{N}}^{j,\zeta_1,\zeta_2}) \\
                & \qquad \le \rho_{0,\RP_1} \limsup\limits_{N \to \infty} \mathds P(\Omega \backslash \Omega_{N,\RP_{N}}^{j,\zeta_1,\zeta_2}) + \eta \ .
            \end{align*}
        \end{proof}
        
        \begin{prop} \label{Proposition 5.3}
            We assume that $\rho \ge \rho_{c,\RP_1}$ and $0 \le \liminf_{N \to \infty} \mathds P(\Omega_{N,\RPN}^{j,\zeta_1,\zeta_2}) \le 1$ for some $0 < \zeta_2<\zeta_1 < 1$ and $2 \le j \in \mathds N$. Then for any $\varepsilon > 0$,
            \begin{align}
                \begin{split}
                    & \liminf\limits_{N \to \infty} \mathds P \left( \dfrac{n_{N,\RPN}^{1,\omega}}{L_N}  \ge \dfrac{1}{j-1} \Big( \liminf_{N \to \infty} \mathds P(\Omega_{N,\RPN}^{j,\zeta_1,\zeta_2}) \rho_{0,\RP_1} - \varepsilon \Big) \right) \\
                    & \qquad \ge \, 1 - \dfrac{1 - \liminf_{N \to \infty} \mathds P(\Omega_{N,\RPN}^{j,\zeta_1,\zeta_2})}{\varepsilon} \rho_{0,\RP_1} \ .
                \end{split}
            \end{align}
        \end{prop}
        \begin{proof}
            For convenience, we define the constant
            \begin{align*}
                c := 1 - \liminf_{N \to \infty} \mathds P(\Omega_{N,\RPN}^{j,\zeta_1,\zeta_2})
            \end{align*}
            and the random variable
            \begin{align*}
                Z_N: \Omega \to \mathds R , \quad \omega \mapsto Z_N^{\omega} := \int\limits_{(0,E_{N,\RPN}^{j-1,\omega}]} \mathcal{B}(E - \mu_{N,\RPN}^{\omega}) \, \mathcal N_{N,\RPN}^{\omega} (\mathrm{d} E) - (1 - c) \rho_{0,V_1}
            \end{align*}
            for each $N \in \mathds N$.
            Due to Proposition~\ref{Proposition 5.2}, $\liminf_{N \to \infty} \E[Z_N^{\omega}] \ge 0$.
            Also, with Lemma~\ref{Lemma 5.1} we conclude $\lim_{N \to \infty} \mathds P (Z_N^{\omega} \le (c+\hat \varepsilon) \rho_{0,V_1}) = 1$ for any $\hat \varepsilon > 0$. In addition,
            we have $Z_N^{\omega} \le \rho$ for $\mathds P$-almost all $\omega \in \Omega$ and all $N \in \mathds N$. Thus, 
            \begin{align*}
                - \hat \varepsilon & \le \E \left[ Z_N^{\omega} \right] 
                = \left( \int\limits_{Z_N^{\omega}<- \varepsilon} + \int\limits_{- \varepsilon \le Z_N^{\omega} \le  (c + \hat\varepsilon)\rho_{0,V_1}} + \int\limits_{Z_N^{\omega} > (c+\hat\varepsilon) \rho_{0,V_1}} \right) Z_N^{\omega}  \, \mathds P(\mathrm{d} \omega) \\
                & \le  - \varepsilon \mathds P(Z_N^{\omega} < - \varepsilon) + (c + \hat \varepsilon) \rho_{0,V_1} + \rho \mathds P \big( Z_N^{\omega} > (c+\hat \varepsilon) \rho_{0,V_1} \big)
            \end{align*}
            for all $\varepsilon, \hat \varepsilon > 0$ and all but finitely many $N \in \mathds N$. Consequently,
            \begin{align*}
                & \limsup\limits_{N \to \infty} \mathds P \Big( \int\limits_{(0,E_{N,\RPN}^{j-1,\omega}]} \mathcal{B}(E - \mu_{N,\RPN}^{\omega})  \, \mathcal N_{N,\RPN}^{\omega}(\mathrm{d} E) < (1-c) \rho_{0,V_1} - \varepsilon \Big) \le \dfrac{(c + \hat \varepsilon) \rho_{0,V_1} + \hat \varepsilon}{\varepsilon} \ .
            \end{align*}
            Since $\hat \varepsilon > 0$ is arbitrary,
            \begin{align*}
                \limsup\limits_{N \to \infty} \mathds P \Big( \int\limits_{(0,E_{N,\RPN}^{j-1,\omega}]} \mathcal{B}(E - \mu_{N,\RPN}^{\omega})  \, \mathcal N_{N,\RPN}^{\omega}(\mathrm{d} E) < (1-c) \rho_{0,V_1} - \varepsilon \Big) \le \dfrac{c }{\varepsilon} \rho_{0,V_1}
            \end{align*}
            for all $\varepsilon > 0$.
            The statement now follows by taking into account that $n_{N,\RPN}^{i,\omega} \le n_{N,\RPN}^{1,\omega}$ $\mathds P$-almost surely for all $N \in \mathds N$ and all $i \in \mathds N$. 
        \end{proof}

        For the next theorem, recall from Section~\ref{SecPrelim} that $\CSU = \CSUleft + \CSUright$ where $\CSUleft, \CSUright > 0$ are such that $[-\CSUleft, \CSUright]$ is the compact support of the single-site potential $u$ of $\RP$. 

        \begin{theorem} \label{Theorem 5.4}
            Let $\RP$ be a Poisson random potential of fixed strength. Then for all $\varepsilon > 0$ there exists a $2 \le j \in \mathds N$ such that for all $0 < \zeta_2<\zeta_1<1$ we have
            \begin{align}
                \liminf\limits_{N \to \infty} \mathds P(\Omega_{N,\RP}^{j,\zeta_1,\zeta_2}) \ge 1 - \varepsilon \ .
            \end{align}
            Furthermore, if there are constants $0 < a \le \CSUright$, $0 < b \le \CSUleft$ such that
            \begin{align} \label{assumption theorem 5.4 second case}
                \max \Big\{2(a+b),  \dfrac{16(8 \pi+1)^2}{\min\big\{ \int_{\CSUright-a}^{\CSUright} u(x) \, \mathrm{d} x, \int_{-\CSUleft}^{-\CSUleft + b} u(x) \, \mathrm{d} x\big\}}\Big\} < c \ ,
            \end{align}
            for an $c > 0$, then
            \begin{align}
                \liminf\limits_{N \to \infty} \mathds P(\Omega_{N,\RP}^{2,\zeta_1,\zeta_2}) \ge \e^{- \nu c}
            \end{align}
            for all $0 < \zeta_2 < \zeta_1 < 1$.
        \end{theorem}
        \begin{proof}
            We start by proving the first part. Let an arbitrary $\varepsilon > 0$ be given. Recall that $\WWS = \min\{ \int_{0}^{\CSUright} u(x) \ud x, \int_{-\CSUleft}^{0} u(x) \ud x\}>0$ is the strength of $V$. We choose a constant $\hat c > 2 \CSU + 16 (8 \pi+1)^2/\WWS$.
            According to Proposition~\ref{Proposition 3.7}, there exists a $j \in \mathds N$ such that
            \begin{align*}
                \liminf\limits_{N \to \infty} \mathds P \left( \llargest - \llargestN{\lceil j/2 \rceil} > \hat c \right) \ge 1 - \varepsilon \ .
            \end{align*}
            Moreover, we have
            \begin{align*}
                \lim\limits_{N \to \infty} \mathds P \Big( E_{N,\RP}^{1,\omega} \le \dfrac{\pi^2}{\big( \llargest - \mathcal C_u \big)^2} \Big) = 1
            \end{align*}
            and
            \begin{align*}
                \lim\limits_{N \to \infty} \mathds P \Bigg( E_{N,\RP}^{j,\omega} \ge \min\Big\{ \dfrac{\pi^2}{(\llargestN{\lceil j/2 \rceil})^2} - \dfrac{(4 \pi)^2(8\pi+1)^2}{\WWS} \dfrac{1}{(\llargestN{\lceil j/2 \rceil} - \mathcal C_u)^2 \llargestN{\lceil j/2 \rceil}} , \dfrac{15}{9} \dfrac{(\nu \pi)^2}{(\ln(L_N))^2} \Big\} \Bigg) = 1 \ ,
            \end{align*}
            see Proposition~\ref{Proposition 4.1}, Lemma~\ref{Lemma 4.2} (with $a = \CSUright$, $b = \CSUleft$), and Proposition~\ref{Proposition 4.4}.
            We assume that $\llargest - \llargestN{\lceil j/2 \rceil} > \hat c$ and choose an $0 < \eta < 1 - 16(8\pi+1)^2/[\WWS(\hat c - 2 \CSU)]$. We conclude that for all sufficiently large $\llargestN{\lceil j/2 \rceil}$ we have
            \begin{align*}
                & \dfrac{\pi^2}{(\llargestN{\lceil j/2 \rceil})^2} - \dfrac{(4 \pi)^2 (8\pi+1)^2}{\WWS} \dfrac{1}{(\llargestN{\lceil j/2 \rceil} - \CSU)^2 \llargestN{\lceil j/2 \rceil}} - \dfrac{\pi^2}{(\llargestN{1} - \CSU)^2} \\
                & \quad \ge \, \dfrac{(1-\eta)\pi^2}{(\llargestN{\lceil j/2 \rceil})^2} - \dfrac{\pi^2}{(\llargestN{1} - \CSU)^2} \\
                & \quad \ge \, \dfrac{(1-\eta)^{-1}\pi^2}{(\llargestN{1})^2} - \dfrac{\pi^2}{(\llargestN{1} - \CSU)^2} \\
                & \quad \ge \, \dfrac{(const.)}{(\llargestN{1})^2}
            \end{align*}
            if $\llargestN{\lceil j/2 \rceil} \le (1-\eta) \llargestN{1}$ and
            \begin{align*}
                & \dfrac{\pi^2}{(\llargestN{\lceil j/2 \rceil})^2} - \dfrac{(4 \pi)^2 (8\pi+1)^2}{\WWS} \dfrac{1}{(\llargestN{\lceil j/2 \rceil} - \CSU)^2 \llargestN{\lceil j/2 \rceil}} - \dfrac{\pi^2}{(\llargestN{1} - \CSU)^2} \\
                & \quad \ge \, \pi^2 \dfrac{(\llargestN{1} - \llargestN{\lceil j/2 \rceil})(\llargestN{1} + \llargestN{\lceil j/2 \rceil}) - 2 \CSU \llargestN{1} + \CSU^2}{(\llargestN{\lceil j/2 \rceil})^2 (\llargestN{1} - \CSU)^2} \\
                & \qquad \qquad - \dfrac{(4 \pi)^2 (8\pi+1)^2}{\WWS} \dfrac{1}{(\llargestN{\lceil j/2 \rceil} - \CSU)^2 \llargestN{\lceil j/2 \rceil}} \\
                & \quad \ge \, \dfrac{\pi^2 (\hat c - 2 \CSU)\llargestN{1}}{(\llargestN{\lceil j/2 \rceil})^2 (\llargestN{1} - \CSU)^2} - \dfrac{(4 \pi)^2 (8\pi+1)^2}{\WWS} \dfrac{1}{(\llargestN{\lceil j/2 \rceil} - \CSU)^2 \llargestN{\lceil j/2 \rceil}} \\
                & \quad \ge \, \dfrac{\pi^2(\hat c - 2 \CSU)}{(\llargestN{\lceil j/2 \rceil})^2 \llargestN{1}} - \dfrac{(4 \pi)^2 (8\pi+1)^2}{\WWS} \dfrac{1}{(\llargestN{\lceil j/2 \rceil} - \CSU)^2 \llargestN{\lceil j/2 \rceil}} \\
                & \quad \ge \, \dfrac{\pi^2(1-\eta)(\hat c - 2 \CSU)}{(\llargestN{\lceil j/2 \rceil})^2 \llargestN{\lceil j/2 \rceil}} - \dfrac{(4 \pi)^2 (8\pi+1)^2}{\WWS} \dfrac{1}{(\llargestN{\lceil j/2 \rceil} - \CSU)^2 \llargestN{\lceil j/2 \rceil}} \\
                & \quad \ge \, \dfrac{(const.)}{(\llargestN{1})^3}
            \end{align*}
            if $\llargestN{\lceil j/2 \rceil} > (1-\eta) \llargestN{1}$.
    
            Secondly, taking into account Proposition~\ref{Proposition 4.1}, Lemma~\ref{Lemma 4.2}, and Proposition~\ref{Proposition 4.3},
            we similarly conclude that if condition \eqref{assumption theorem 5.4 second case} is fulfilled and if $\llargest - \llargestzwei > c$, then we have, with $\WWS^{a,b} := \min\big\{ \int_{\CSUright-a}^{\CSUright} u(x) \, \mathrm{d} x, \int_{-\CSUleft}^{-\CSUleft + b} u(x) \, \mathrm{d} x\big\}$ and any $0 < \hat \eta < 1 - 16 (8\pi+1)^2/(\WWS^{a,b} c)$, for all sufficiently large $\llargestN{2}$
            \begin{align*}
                & \dfrac{\pi^2}{(\llargestN{2} - (\CSU - a - b) )^2} - \dfrac{(4 \pi)^2 (8\pi+1)^2}{\WWS^{a,b}} \dfrac{1}{(\llargestN{2} - \CSU)^2 \llargestN{2}} - \dfrac{\pi^2}{(\llargestN{1} - \CSU)^2} \ge \dfrac{(const.)}{(\llargestN{1})^2}
            \end{align*}
            if $\llargestN{2} \le (1-\hat \eta) \llargestN{1}$, and otherwise
            \begin{align*}
                & \dfrac{\pi^2}{(\llargestN{2} - (\CSU - a - b))^2} - \dfrac{(4 \pi)^2 (8\pi+1)^2}{\WWS^{a,b}} \dfrac{1}{(\llargestN{2} - \CSU)^2 \llargestN{2}} - \dfrac{\pi^2}{(\llargestN{1} - \CSU)^2} \\
                & \quad \ge \, \pi^2 \dfrac{(\llargestN{1} - \llargestN{2})(\llargestN{1} + \llargestN{2}) - 2 \CSU (\llargestN{1} - \llargestN{2}) - 2 (a+b) \llargestN{2}}{(\llargestN{2} - (\CSU - a - b))^2 (\llargestN{1} - \CSU)^2} \\
                & \qquad \qquad - \dfrac{(4 \pi)^2 (8\pi+1)^2}{\WWS^{a,b}} \dfrac{1}{(\llargestN{2} - \CSU)^2 \llargestN{2}} \\
                & \quad \ge \, \dfrac{\pi^2 c \llargestN{1}}{(\llargestN{2} - (\CSU - a - b))^2 (\llargestN{1} - \CSU)^2} - \dfrac{(4 \pi)^2 (8\pi+1)^2}{\WWS^{a,b}} \dfrac{1}{(\llargestN{2} - \CSU)^2 \llargestN{2}} \\
                & \quad \ge \, \dfrac{\pi^2 c (1-\hat \eta)}{(\llargestN{2} - (\CSU - a - b))^2 \llargestN{2}} - \dfrac{(4 \pi)^2 (8\pi+1)^2}{\WWS^{a,b}} \dfrac{1}{(\llargestN{2} - \CSU)^2 \llargestN{2}} \\
                & \quad \ge \, \dfrac{(const.)}{(\llargestN{1})^3} \ .
            \end{align*}

            Lastly, the claims of this theorem now follow since
            \begin{align*}
                \lim\limits_{N \to \infty} \mathds P \left( (1-\zeta) \nu^{-1} \ln N \le \llargest \le (1+\zeta) \nu^{-1} \ln N \right) = 1
            \end{align*}
            for any $\zeta > 0$,
            \begin{align*}
                \liminf\limits_{N \to \infty} \mathds P \left( \llargest - \llargestzwei > c \right) \ge \e^{-\nu c} \ ,
            \end{align*}
            see Lemma~\ref{Lemma 3.2} and Proposition~\ref{Proposition 3.4}, and since $\lim_{N \to \infty} \mathds P(\llargestN{k} > \tilde c) = 1$ for any $k \in \mathds N$ and any $\tilde c > 0$.
        \end{proof}

        \begin{theorem} \label{Theorem 5.5}
            Let $(V_N)_{N \in \mathds N}$ be a Poisson random potential with a strength that converges to infinity. Then for all $0 < \zeta_2 < \zeta_1 < 1$ one has
            \begin{align}
                \lim\limits_{N \to \infty} \mathds P(\Omega_{N,\RPN}^{2,\zeta_1,\zeta_2}) = 1 \ .
            \end{align}
        \end{theorem}
        \begin{proof}
            Due to our assumptions, there exists a sequence $(a_N)_{N \in \mathds N} \subset (0,\infty)$ that converges to zero but slowly enough such that
            \begin{align*}
                \lim\limits_{N \to \infty} \WWS_N \min\Big\{ \int_{\CSUright- a_N}^{\CSUright} u(x) \ud x, \int_{-\CSUleft}^{-\CSUleft + a_N} u(x) \ud x \Big\} = \infty \ .
            \end{align*}
            Using Lemma~\ref{Lemma 4.2}, we proceed as in the proof of the second part of Theorem~\ref{Theorem 5.4}
            but instead of $c$ with a sequence $(c_N)_{N \in \mathds N}$ that converges to zero sufficiently slowly and for which
            \begin{align*}
                c_N > \max\Big\{ 4 a_N, \dfrac{16 (8\pi+1)^2}{\widetilde S_N^{a_N,a_N}} \Big\}
            \end{align*}
            for all $N \in \mathds N$ holds as well as by using Proposition~\ref{Proposition 3.5} instead of Proposition~\ref{Proposition 3.4}.
        \end{proof}

        We finally prove the occurrence of type-I g-BEC. Recall that in this work we assume a single-site potential $u$ that is a nonnegative, compactly supported, and bounded measurable function. In Corollaries~\ref{Corollary 1} and~\ref{Corollary 2}, we assume a Poisson random potential of fixed strength. Here, we prove that if the particle density is larger than the critical density, then type-I g-BEC occurs with probability almost one. In addition, we also show that if the particle density is larger than the critical one and if the strength of the random potential is sufficiently large in a certain sense, then a type-I g-BEC in which only the ground state is macroscopically occupied occurs with a probability arbitrarily close to one. Note that in the former case no assumption regarding the strength of the random potential is made. In the subsequent Corollary~\ref{Corollary 3} we show that a type-I g-BEC occurs in probability, and consequently also in the $r$th mean, $r \ge 1$, if the strength of the Poisson random potential converges to infinity in the thermodynamic limit, in the sense of \eqref{def RP converges to infty zwei} but otherwise arbitrarily slowly, and if the particle density is sufficiently large. At this point, we would also like to remind the reader of Remark~\ref{Remark LS with finite strength}.
        
        \begin{cor}[Type-I g-BEC with probability almost one] \label{Corollary 1}
         Let $V$ be a Poisson random potential of fixed strength and $\rho > \rho_{c,V}$. Then for all $\varepsilon>0$ there exists a $2 \le j \in \mathds N$ such that for all $\eta > 0$,
                \begin{align}
                    \liminf\limits_{N \to \infty} \mathds P \left(\dfrac{n_{N,\RP}^{1,\omega}}{N} \ge \dfrac{1 - \varepsilon}{j-1} \dfrac{\rho_{0,V}}{\rho},\dfrac{n_{N,\RP}^{j,\omega}}{N} < \eta \right) \ge 1 - \varepsilon \ .
                \end{align}
        \end{cor}
        \begin{proof}
            Let $0 < \zeta_2<\zeta_1<0$ be arbitrarily given. Using the first part of Theorem~\ref{Theorem 5.4} we conclude that for all $\varepsilon > 0$ there exists a $2 \le j \in \mathds N$ such that
            \begin{align*}
                \liminf\limits_{N \to \infty} \mathds P \big( \Omega_{N,\RP}^{j,\zeta_1,\zeta_2} \big) & \ge 1 - \varepsilon \ .
            \end{align*}
            By Theorem~\ref{Proposition 5.3} we thus have
            \begin{align*}
                \liminf\limits_{N \to \infty} \mathds P \left( \dfrac{n_{N,\RP}^{1,\omega}}{L_N}  \ge \dfrac{1 - \sqrt{\varepsilon} - \varepsilon}{j-1} \rho_{0,V} \right) \ge 1 - \sqrt{\varepsilon} \ .
            \end{align*}
            In addition, for any $\eta > 0$, for all but finitely many $N \in \mathds N$, and for all $\omega \in \Omega_{N,\RP}^{j,\zeta_1,\zeta_2}$,
            \begin{align*}
                \dfrac{n_{N,\RPN}^{j,\omega}}{N} & = \dfrac{1}{N} \left( \mathrm{e}^{\beta (E_{N,\RP}^{j,\omega} - \mu_{N,\RP}^{\omega})} - 1 \right)^{-1} \le \dfrac{1}{N} \left( \mathrm{e}^{\beta (E_{N,\RP}^{j,\omega} - E_{N,\RP}^{1,\omega})} - 1 \right)^{-1} \\
                & \le \dfrac{1}{N} \Big( \beta (E_{N,\RP}^{j,\omega} - E_{N,\RP}^{1,\omega}) \Big)^{-1} 
                < \eta \ .
            \end{align*}
            Consequently,
            \begin{align}
                \liminf\limits_{N \to \infty} \mathds P \left( \dfrac{n_{N,\RP}^{j,\omega}}{N}  < \eta \right) \ge 1 - \varepsilon
            \end{align}
            and
            \begin{align}
                \liminf\limits_{N \to \infty} \mathds P \left( \dfrac{n_{N,\RP}^{1,\omega}}{L_N}  \ge \dfrac{ 1 - \sqrt{\varepsilon} - \varepsilon}{j-1} \rho_{0,V}, \dfrac{n_{N,\RP}^{j,\omega}}{N}  < \eta \right) \ge 1 - \sqrt{\varepsilon} - \varepsilon
            \end{align}
            for any $\eta > 0$.
            \end{proof}

        \begin{cor}\label{Corollary 2}
            Let $V$ be a Poisson random potential of fixed strength and with a single-site potential $u$ that has the compact support $[-\CSUleft, \CSUright]$. Let $\rho > \rho_{c,V}$. Then we have:
                If there are $0 < a \le \CSUright$, $0 < b \le \CSUleft$ such that
                \begin{align}
                    \max \Bigg\{ 2(a+b),  \dfrac{16(8\pi+1)^2}{\min\big\{ \int_{\CSUright-a}^{\CSUright} u(x) \, \mathrm{d} x, \int_{-\CSUleft}^{-\CSUleft + b} u(x) \, \mathrm{d} x\big\}} \Bigg\} < \dfrac{1}{\nu} \ln[(1- \varepsilon)^{-1}] \ .
                \end{align}
                for an $0 < \varepsilon < 1$, then for all $\eta > 0$
                \begin{align}
                    \liminf\limits_{N \to \infty} \mathds P \left(\dfrac{n_{N,\RP}^{1,\omega}}{N} \ge ( 1 - \sqrt{\varepsilon} - \varepsilon ) \dfrac{\rho_{0,V}}{\rho},\dfrac{n_{N,\RP}^{2,\omega}}{N} < \eta \right) \ge 1 - \sqrt{\varepsilon} - \varepsilon \ .
                \end{align}
        \end{cor}
        \begin{proof}
         To prove this corollary, one proceeds very similarly as in the proof of Corollary~\ref{Corollary 1} but uses the second rather than the first part of Theorem~\ref{Theorem 5.4}.

        \end{proof}
	
        \begin{cor}[Type-I g-BEC in probability]\label{Corollary 3}
            Let $(V_N)_{N \in \mathds N}$ be a Poisson random potential with a strength that converges to infinity. Furthermore, let $\rho > \rho_{c,\RP_1}$. Then a type-I g-BEC where only the one-particle ground state is macroscopically occupied occurs in probability,
            \begin{equation*}
                \lim\limits_{N \to \infty} \mathds P \left(\left| \dfrac{n_{N,\RPN}^{1,\omega}}{N} - \dfrac{\rho_{0,\RP_1}}{\rho} \right| < \eta,\dfrac{n_{N,\RPN}^{2,\omega}}{N} < \eta' \right) =1
            \end{equation*}
            for all $\eta,\eta' > 0$.
        \end{cor}
        \begin{proof}
            Let $0 < \zeta_2<\zeta_1<0$ be arbitrarily given. With Lemma~\ref{Lemma 5.1} and Theorems~\ref{Proposition 5.3} and~\ref{Theorem 5.5}, we conclude
            \begin{align*}
                \lim\limits_{N \to \infty} \mathds P \left(\left| \dfrac{n_{N,\RPN}^{1,\omega}}{L_N} - \rho_{0,\RP_1} \right| < \eta \right) = 1
            \end{align*}
            for all $\eta > 0$.
            Consequently,
            \begin{align*}
                & \lim\limits_{N \to \infty} \mathds P \left( \left| \dfrac{n_{N,\RPN}^{1,\omega}}{L_N} - \rho_{0,\RP_1} \right| < \eta , \dfrac{n_{N,\RPN}^{2,\omega}}{N} < \eta' \right) \\
                & \quad \ge \, \lim\limits_{N \to \infty} \mathds P \left( \left| \dfrac{n_{N,\RPN}^{1,\omega}}{L_N} - \rho_{0,\RP_1} \right| < \eta \right) + \lim\limits_{N \to \infty} \mathds P \left( \dfrac{n_{N,\RPN}^{2,\omega}}{N} < \eta' \right) - 1 = 1
            \end{align*}
            for all $\eta, \eta' > 0$.
        \end{proof}
        
        Under the same assumptions as in Corollary~\ref{Corollary 3}, a type-I g-BEC also occurs in the $r$th mean, $r \ge 1$, where only the ground state is macroscopically occupied, that is, one has
        \begin{align}
            \lim\limits_{N \to \infty} \E \, \left| \dfrac{n_{N,\RPN}^{1,\omega}}{N} - \dfrac{\rho_{0,\RP_1}}{\rho} \right|^r = 0 \ \text{ and } \ \lim\limits_{N \to \infty} \E \left| \dfrac{n_{N,\RPN}^{j,\omega}}{N} \right|^r = 0 
        \end{align}
        for all $2 \le j \in \mathds N$, see also \cite[Corollary 2.11]{KPSConditionTypeOneBEC2020}. This follows with Corollary~\ref{Corollary 3} and standard results from measure theory. In particular, note that $|n_{N,\RPN}^{1,\omega}/{N}- \rho_{0,\RP_1}/\rho| \le 1$ as well as $|n_{N,\RPN}^{j,\omega}/N|\le 1$ for all $2 \le j \in \mathds N$, for $\mathds P$-almost all $\omega \in \Omega$, and for all $N \in \mathds N$.
        
        Lastly, we show that in all of the cases of Corollaries~\ref{Corollary 1} - \ref{Corollary 3} the one-particle ground state is $\mathds P$-almost surely macroscopically occupied, in the sense of Definition~\ref{Definition makroskopische Besetzung}.
	
        \begin{cor}[Almost sure macroscopic occupation of the ground state]\label{Corollary 5.9}
            Let $(\RP_N)_{N \in \mathds N}$ be a Poisson random potential of fixed strength or with a strength that converges to infinity, and $\rho > \rho_{c,\RP_1}$. Then
            \begin{align}
                \mathds P \Big( \limsup\limits_{N \to \infty} \dfrac{n_{N,\RPN}^{1,\omega}}{N}  > 0 \Big) = 1 \ .
            \end{align}
        \end{cor}
        \begin{proof}
            We follow the proof idea of \cite[Theorem 3.5]{KPS182}. Firstly, if there is an $\varepsilon > 0$ such that $\mathds P(\lim_{N \to \infty} n_{N,\RPN}^{1,\omega} / N = 0) > \sqrt{\varepsilon}$, then there exists an $\varepsilon > 0$ such that for any $2 \le j \in \mathds N$,
            \begin{align*}
            \begin{split}
                & \limsup\limits_{N \to \infty} \mathds{P} \Big( \dfrac{n_{N,\RPN}^{1,\omega}}{N} \ge \dfrac{1}{j-1} (1 - \sqrt{\varepsilon} - \varepsilon) \rho_{0,V_1} \Big) \\
                & \quad \le \mathds P \Big( \limsup\limits_{N \to \infty} \dfrac{n_{N,\RPN}^{1,\omega}}{N} \ge \dfrac{1}{j-1} (1 - \sqrt{\varepsilon} - \varepsilon) \rho_{0,V_1} \Big) \\
                & \quad \le \mathds P \Big( \limsup\limits_{N \to \infty} \dfrac{n_{N,\RPN}^{1,\omega}}{N} > 0 \Big)\\
                & \quad < 1 - \sqrt{\varepsilon} \ .
            \end{split}
            \end{align*}
            Secondly, by Proposition~\ref{Proposition 5.3} as well as Theorem~\ref{Theorem 5.4} and Theorem~\ref{Theorem 5.5}, respectively, the following holds: For all $\varepsilon > 0$ there is a $2 \le j \in \mathds N$ such that
            \begin{align*}
                & \limsup\limits_{N \to \infty} \mathds P \Big( \dfrac{n_{N,\RPN}^{1,\omega}}{N} \ge \dfrac{1}{j-1} (1 - \sqrt{\varepsilon} - \varepsilon) \rho_{0,V_1} \Big) \ge 1 - \sqrt{\varepsilon} \ .
            \end{align*}
        \end{proof}

    \appendix
    %
    \section{Generalized Bose--Einstein condensation}\label{secAppendix}

        In this appendix, we prove Theorem~\ref{Theorem g-BEC}, that is, the $\mathds P$-almost sure occurrence of g-BEC in the case of a Poisson random potential with a strength that converges to infinity, see Theorem~\ref{Theorem A.7}. For this, we firstly state several lemmata (Lemmata~\ref{Lemma A.1} - \ref{Lemma A.6}). For the proof of Lemma~\ref{Lemma A.1}, we refer to \cite[Corollary 3.5]{lenoble2004bose} in combination with \cite[Lemma A.5]{KPS182}.
        We prove Lemmata~\ref{Lemma A.2} - \ref{Lemma A.6}, since these statements are more general than the ones that can be found in \cite{lenoble2004bose} and \cite{KPS182}. Also, note that we use these Lemmata in our proofs in Section~\ref{secMainResults}.

        Recall Remark~\ref{Remark typical omega} and the fact that the chemical potential $\mu_{N,\RPN}^{\omega} \in (-\infty, E_{N,\RPN}^{1,\omega})$ is for $\mathds P$-almost all $\omega \in \Omega$ and for all $N \in \mathds N$ uniquely determined by the equation \eqref{Gleichung Bedingung mu aequivalent}.

        \begin{lemma} \label{Lemma A.1}
            Let $V$ be a Poisson random potential of fixed strength. If $\mu < 0$, then $\mathds P$-almost surely
            \begin{align}
                \lim\limits_{N \to \infty} \int\limits_{(0,\infty)} \mathcal B(E - \mu)  \, \mathcal N_{N,V}^{\omega}(\mathrm{d} E) = \int\limits_{(0,\infty)} \mathcal B(E - \mu)  \, \mathcal N_{\infty,V}(\mathrm{d} E) \ .
            \end{align}
        \end{lemma}

        For the proofs of the next two lemmata, we define 
        \begin{align} \label{definition cutofffunction eins}
            \cutofffunctioneins{\mathcal E_1}{\mathcal E_2} : \mathds R \to [0,1], \ E \mapsto \cutofffunctioneins{\mathcal E_1}{\mathcal E_2}(E):=
            \begin{cases}
                0 \quad & \text{ if } E \le \mathcal E_1/2 \\
                \dfrac{E - \mathcal E_1/2}{\mathcal E_1/2} & \text{ if } \mathcal E_1/2 < E <\mathcal E_1 \\
                1 \quad & \text{ if } \mathcal E_1 \le E \le \mathcal E_2 \\
                1 - (E - \mathcal E_2) \quad & \text{ if } \mathcal E_2 < E < \mathcal E_2 + 1\\
                0 \quad & \text{ if } E \ge \mathcal E_2 + 1
            \end{cases}
        \end{align}
        where $\mathcal E_1, \mathcal E_2 > 0$ are two constants with $0<\mathcal E_1<\mathcal E_2$. The reason for the introduction of this continuous cutoff function is the $\mathds P$-almost sure convergence of the density of states $(\mathcal N_{N,V}^{\omega})_{N \in \mathds N}$ in the vague rather than weak sense.

        \begin{lemma} \label{Lemma A.2}
            Let $V$ be a Poisson random potential of fixed strength. Furthermore, let $\varepsilon>0$ and $\mathcal E_2 > \varepsilon$ be given, and $\cutofffunctioneins{\varepsilon}{\mathcal E_2}$ as in \eqref{definition cutofffunction eins}. Then for $\mathds P$-almost all $\omega \in \Omega$ we have: If $(\mugeneral_N^{\omega})_{N \in \mathds N} \subset \mathds R$ is a sequence with $\lim_{N \to \infty} \mugeneral_N^{\omega} = 0$, then
            \begin{align} \label{lim int NN int N infty}
                & \lim\limits_{N \to \infty} \int\limits_{\mathds R} \cutofffunctioneins{\varepsilon}{\mathcal E_2}(E) \, \mathcal B(E - \mugeneral_N^{\omega})  \, \mathcal N_{N,V}^{\omega} (\mathrm{d} E) = \int\limits_{\mathds R} \cutofffunctioneins{\varepsilon}{\mathcal E_2}(E) \, \mathcal B(E)  \, \mathcal N_{\infty,V} (\mathrm{d} E) \ .
            \end{align}
        \end{lemma}
        \begin{proof}
            For $\mathds P$-almost all $\omega \in \Omega$ and all $N \in \mathds N$ we have
            \begin{align*}
                & \left|\, \int\limits_{\mathds R} \cutofffunctioneins{\varepsilon}{\mathcal E_2}(E) \, \mathcal B(E - \mugeneral_N^{\omega})  \, \mathcal N_{N,V}^{\omega} (\mathrm{d} E) - \int\limits_{\mathds R} \cutofffunctioneins{\varepsilon}{\mathcal E_2}(E) \, \mathcal B(E)  \, \mathcal N_{\infty,V} (\mathrm{d} E) \, \right| \\
                &\quad \le \, \left| \, \int\limits_{\mathds R} \cutofffunctioneins{\varepsilon}{\mathcal E_2}(E) \big[ \mathcal B(E - \mugeneral_N^{\omega}) - \mathcal B(E) \big] \, \mathcal N_{N,V}^{\omega}(\mathrm{d} E) \, \right| \\
                & \qquad + \, \left|\, \int\limits_{\mathds R} \cutofffunctioneins{\varepsilon}{\mathcal E_2}(E) \mathcal B(E) \, \mathcal N_{N,\RP}^{\omega}(\mathrm{d} E) - \int\limits_{\mathds R} \cutofffunctioneins{\varepsilon}{\mathcal E_2}(E) \mathcal B(E) \, \mathcal N_{\infty,V}(\mathrm{d} E) \, \right| \ .
            \end{align*}
            The last term $\mathds P$-almost surely converges to zero in the limit $N \to \infty$, because $\mathcal N_{N,V}^{\omega}$ $\mathds P$-almost surely convergences in the vague sense to $\mathcal N_{\infty,V}$ and $\cutofffunctioneins{\varepsilon}{\mathcal E_2} \, \mathcal B$ is a continuous, compactly supported function. In addition, using that
            $$\mathcal B(E - \mugeneral_N^{\omega}) - \mathcal B(E) = (\e^{\beta E} - \e^{\beta(E - \mugeneral_N^{\omega})}) \mathcal B(E ) \mathcal B(E - \mugeneral_N^{\omega} )$$
            and
                $$\limsup\limits_{N \to \infty} \int\limits_{\mathds R} \cutofffunctioneins{\varepsilon}{\mathcal E_2}(E) \mathcal B(E - \mugeneral_N^{\omega} ) \, \mathcal N_{N,V}^{\omega}(\mathrm{d} E) \le \mathcal B(\varepsilon/4) \mathcal N_{\infty,V}^{\mathrm{I}}(2 \mathcal E_2 + 2)$$
            for $\mathds P$-almost all $\omega \in \Omega$, see also \eqref{upper bound NNI}, we conclude
            \begin{align*}
                & \lim\limits_{N \to \infty} \left| \, \int\limits_{\mathds R} \cutofffunctioneins{\varepsilon}{\mathcal E_2}(E) \big[ \mathcal B(E - \mugeneral_N^{\omega}) - \mathcal B(E) \big]  \, \mathcal N_{N,V}^{\omega}(\mathrm{d} E) \, \right| \\
                & \quad \le \lim\limits_{N \to \infty} \e^{\beta (\mathcal E_2 + 1) } \mathcal B(\varepsilon/2) \mathcal B(\varepsilon/4) \mathcal N_{\infty,V}^{\mathrm{I}}(2 \mathcal E_2 + 2) \, \big| 1 - \e^{-\beta \mugeneral_N^{\omega}} \big|\ ,
            \end{align*}
            which $\mathds P$-almost surely converges to zero as well.
        \end{proof}

        \begin{lemma} \label{Lemma A.3}
            Suppose $V$ is a Poisson random potential of fixed strength. Let an $\varepsilon > 0$ and an $\omega \in \widetilde \Omega$ be arbitrarily given, and let $(\mugeneral_N^{\omega})_{N \in \mathds N} \subset \mathds R$ be a sequence that converges to zero. Then
            \begin{align}
                \limsup\limits_{N \to \infty} \int\limits\limits_{(\varepsilon,\infty)} \mathcal B(E - \mugeneral_N^{\omega})  \, \mathcal N_{N,V}^{\omega} (\mathrm{d} E) & \le \int\limits\limits_{(\varepsilon,\infty)} \mathcal B(E) \, \mathcal N_{\infty,V}(\mathrm{d} E) + \dfrac{2}{\beta \varepsilon} \mathcal N_{\infty,V}^{\mathrm{I}}(\varepsilon) \label{Lemma 2 Gen BEC HG 1} \\
                \shortintertext{and}
                \liminf\limits_{N \to \infty} \int\limits\limits_{(\varepsilon,\infty)} \mathcal B(E - \mugeneral_N^{\omega})  \, \mathcal N_{N,V}^{\omega} (\mathrm{d} E) & \ge \int\limits\limits_{(\varepsilon,\infty)} \mathcal B(E) \, \mathcal N_{\infty,V}(\mathrm{d} E) - \dfrac{4}{\beta \varepsilon} \mathcal N_{\infty,V}^{\mathrm{I}}(2\varepsilon) \ . \label{Lemma 2 Gen BEC HG 2}
            \end{align}
        \end{lemma}
        \begin{proof}
  
            We begin by showing \eqref{Lemma 2 Gen BEC HG 1}. Let $E_2 > \varepsilon$ be arbitrarily given. Then
            \begin{align*}
                & \limsup\limits_{N \to \infty} \int\limits\limits_{(\varepsilon,\infty)} \mathcal B(E - \mugeneral_N^{\omega})  \, \mathcal N_{N,V}^{\omega} (\mathrm{d} E) \\
                & \quad \le \limsup\limits_{N \to \infty} \int\limits_{(\varepsilon,E_2]} \mathcal B(E - \mugeneral_N^{\omega}) \, \mathcal N_{N,V}^{\omega} (\mathrm{d} E) + \limsup\limits_{N \to \infty} \int\limits_{(E_2,\infty)} \mathcal B(E - \mugeneral_N^{\omega})  \, \mathcal N_{N,V}^{\omega} (\mathrm{d} E) \ .
            \end{align*}
            Using the fact that the function $\mathcal B$ is monotonically decreasing, integration by parts, and the inequality $\mathcal N_{N,\RP}^{\mathrm{I},\omega}(E) \le \pi^{-1} E^{1/2}$ for $\mathds P$-almost all $\omega \in \Omega$, for all $E \ge 0$, and all $N \in \mathds N$ (see, e.g., \cite[Remark 3.2.2]{pechmanndiss} for more details regarding the last inequality), we conclude
            \begin{align*}
                & \int\limits_{(E_2,\infty)} \mathcal B(E - \mugeneral_N^{\omega})  \, \mathcal N_{N,V}^{\omega} (\mathrm{d} E) \\
                & \quad \le \lim_{E_3 \to \infty} \Big[ \mathcal B(E_3 - \varepsilon/2) \,\mathcal N_{N,V}^{\mathrm{I},\omega} (2E_3) + \beta \int_{E_2}^{E_3} \mathcal N_{N,V}^{\mathrm{I}, \omega} ( E) \big[ \mathcal B(E- \varepsilon/2) \big]^{2} \e^{\beta(E - \varepsilon/2)} \, \mathrm{d} E \Big] \\
                & \quad \le \beta\pi^{-1} \int_{E_2}^{\infty} E^{1/2} \big[ \mathcal B(E- \varepsilon/2) \big]^{2} \e^{\beta(E - \varepsilon/2)} \, \mathrm{d} E
            \end{align*}
            for all but finitely many $N \in \mathds N$, which converges to zero in the limit $E_2 \to \infty$.
            Thus, with $\cutofffunctioneins{\varepsilon}{E_2}$ as in \eqref{definition cutofffunction eins}, with Lemma~\ref{Lemma A.2}, and since $\mathcal B(\varepsilon/2) \le 2/(\beta \varepsilon)$,
            \begin{align*}
                & \limsup\limits_{N \to \infty} \int\limits_{(\varepsilon,\infty)} \mathcal B(E - \mugeneral_N^{\omega})  \, \mathcal N_{N,V}^{\omega} (\mathrm{d} E) \le \lim\limits_{E_2 \to \infty} \limsup\limits_{N \to \infty} \int\limits_{(\varepsilon,E_2]} \mathcal B(E - \mugeneral_N^{\omega}) \, \mathcal N_{N,V}^{\omega} (\mathrm{d} E) \\
                & \quad \le \lim\limits_{E_2 \to \infty} \int\limits_{\mathds R} \cutofffunctioneins{\varepsilon}{E_2}(E) \, \mathcal B(E) \, \mathcal N_{\infty,V}(\mathrm{d} E) \le \int\limits\limits_{(\varepsilon,\infty)} \mathcal B(E) \, \mathcal N_{\infty,V}(\mathrm{d} E) + \dfrac{2}{\beta \varepsilon} \mathcal N_{\infty,V}^{\mathrm{I}}(\varepsilon) \ .
            \end{align*}

            Next, we have
            \begin{align*}
                \liminf\limits_{N \to \infty} \int\limits\limits_{(\varepsilon,\infty)} \mathcal B(E - \mugeneral_N^{\omega})  \, \mathcal N_{N,V}^{\omega} (\mathrm{d} E) & \ge \liminf\limits_{N \to \infty} \int\limits_{\mathds R} \cutofffunctioneins{\varepsilon}{E_2}(E) \, \mathcal B(E - \mugeneral_N^{\omega}) \, \mathcal N_{N,V}^{\omega} (\mathrm{d} E) \\
                & \quad - \, \limsup\limits_{N \to \infty} \int\limits_{[\varepsilon/2,\varepsilon]} \cutofffunctioneins{\varepsilon}{E_2}(E) \, \mathcal B(E - \mugeneral_N^{\omega})  \, \mathcal N_{N,V}^{\omega} (\mathrm{d} E)
            \end{align*}
            for all $E_2 > \varepsilon$. Inequality \eqref{Lemma 2 Gen BEC HG 2} now follows after using Lemma~\ref{Lemma A.2}, inequality \eqref{upper bound NNI}, and the fact that $\mathcal B(E - \mugeneral_N^{\omega}) \le \mathcal B(\varepsilon/2 - \varepsilon/4)$ for all $E \ge \varepsilon/2$ and all but finitely many $N \in \mathds N$.
        \end{proof}

        \begin{lemma} \label{Lemma A.4}
            Let $V$ be a Poisson random potential of fixed strength. Moreover, let $\omega \in \widetilde \Omega$ and a sequence $(\mugeneral_N^{\omega})_{N \in \mathds N} \subset \mathds R$ be given. If $(\mugeneral_N^{\omega})_{N \in \mathds N}$ converges to zero, then we have
            \begin{align}
                \begin{split}
                    \lim\limits_{\varepsilon \searrow 0} \limsup\limits_{N \to \infty} \int\limits\limits_{(\varepsilon,\infty)} \mathcal B(E - \mugeneral_N^{\omega})  \, \mathcal N_{N,V}^{\omega}(\mathrm{d} E) & = \lim\limits_{\varepsilon \searrow 0} \liminf\limits_{N \to \infty} \int\limits\limits_{(\varepsilon,\infty)} \mathcal B(E - \mugeneral_N^{\omega})  \, \mathcal N_{N,V}^{\omega}(\mathrm{d} E) \\
                    & = \rho_{c,\RP} \ .
                \end{split}
            \end{align}
        \end{lemma}
        \begin{proof}
            Since $\rho_{c,\RP} < \infty$, we conclude with Definition~\eqref{Definition critical density} and Lemma~\ref{Lemma A.3} that
            \begin{align*}
                \rho_{c,\RP} & = \lim\limits_{\varepsilon \searrow 0} \int\limits\limits_{(\varepsilon,\infty)} \mathcal B(E)  \, \mathcal N_{\infty,V}(\mathrm{d} E) = \lim\limits_{\varepsilon \searrow 0} \limsup\limits_{N \to \infty} \int\limits\limits_{(\varepsilon,\infty)} \mathcal B(E - \mugeneral_N^{\omega})  \, \mathcal N_{N,V}^{\omega}(\mathrm{d} E) \\
                & = \lim\limits_{\varepsilon \searrow 0} \liminf\limits_{N \to \infty} \int\limits\limits_{(\varepsilon,\infty)} \mathcal B(E - \mugeneral_N^{\omega})  \, \mathcal N_{N,V}^{\omega}(\mathrm{d} E) \ .
            \end{align*}
        \end{proof}

        \begin{lemma} \label{Lemma A.5}
            Let $(V_N)_{N \in \mathds N}$ be a Poisson random potential of fixed strength or with a strength that converges to infinity, and $\omega \in \widetilde \Omega$. Then the sequence $( \mu_{N,\RPN}^{\omega} )_{N \in \mathds N}$ has at least one accumulation point, and all accumulation points of $( \mu_{N,\RPN}^{\omega} )_{N \in \mathds N}$ are smaller than or equal to zero.
        \end{lemma}

        \begin{proof}
            We follow \cite[Theorem 4.1]{lenoble2004bose} in large parts. For all $N \in \mathds N$ and $\omega \in \widetilde \Omega$, we define $\Phi_N^{\omega} := |\Lambda_N|^{-1} \sum_{j \in \mathds N} \e^{-\beta E_{N,\RPN}^{j, \omega}}$. 
            Let $N \in \mathds N$ and $\omega \in \widetilde \Omega$ be arbitrarily given. Using integration by parts and the fact that $\mathcal N_{N,\RPN}^{\mathrm{I},\omega}(E) \le \pi^{-1} E^{1/2}$ for $\mathds P$-almost all $\omega \in \Omega$, for all $E \ge 0$, and all $N \in \mathds N$,
            we obtain
            \begin{align*}
                0 < \Phi_N^{\omega} & = \int\limits_{(0,\infty)} \e^{-\beta E} \,\mathcal N_{N,\RPN}^{\omega}(\mathrm{d} E)
                \le \beta \int_{0}^{\infty} \pi^{-1} E^{1/2} \e^{-\beta E} \, \mathrm{d} E < \infty \ .
            \end{align*}
            We conclude
            \begin{align*}
                \rho & = \dfrac{1}{|\Lambda_N|} \sum\limits_{j \in \mathds N} \Big( \e^{\beta ( E_{N,\RPN}^{j,\omega} - \mu_{N,\RPN}^{\omega})} -1 \Big)^{-1} = \dfrac{1}{|\Lambda_N|} \sum\limits_{j \in \mathds N} \dfrac{\e^{-\beta E_{N,\RPN}^{j,\omega}}}{\e^{-\beta \mu_{N,\RPN}^{\omega}} - \e^{-\beta E_{N,\RPN}^{j,\omega}} } \\
                & \le \dfrac{1}{\e^{-\beta \mu_{N,\RPN}^{\omega}} - \e^{-\beta E_{N,\RPN}^{1,\omega}}} \dfrac{1}{|\Lambda_N|} \sum\limits_{j \in \mathds N} \e^{- \beta E_{N,\RPN}^{j,\omega}} = \dfrac{\e^{\beta \mu_{N,\RPN}^{\omega}}}{1 - \e^{-\beta (E_{N,\RPN}^{1,\omega} - \mu_{N,\RPN}^{\omega})}} \Phi_N^{\omega} 
            \end{align*}
            and
            \begin{align*}
                \beta^{-1} \ln \Bigg( \dfrac{\rho}{\Phi_N^{\omega} + \rho \e^{-\beta E_{N,\RPN}^{1,\omega}}} \Bigg) & \le \mu_{N,\RPN}^{\omega} \ .
            \end{align*}
            Therefore, we $\mathds P$-almost surely obtain
            \begin{align*}
                \beta^{-1} \ln \Bigg( \dfrac{\rho}{\beta \int_{0}^{\infty} \pi^{-1} E^{1/2} \e^{-\beta E} \, \mathrm{d} E + \rho } \Bigg) \le \liminf\limits_{N \to \infty} \mu_{N,\RPN}^{\omega}\le \limsup\limits_{N \to \infty} \mu_{N,\RPN}^{\omega} \le 0 \ .
            \end{align*}
           The last step follows from the facts that the ground-state energy $E_{N,\RPN}^{1,\omega}$ $\mathds P$-almost surely converges to zero, see Proposition~\ref{Proposition 4.1} and an appropriate version of Lemma~\ref{Lemma 3.2} (in combination with the Borel-Cantelli lemma), and that $\mu_{N,\RPN}^{\omega}< E_{N,\RPN}^{1,\omega}$ for $\mathds P$-almost all $\omega \in \Omega$ and for all $N \in \mathds N$. Thus, the sequence $( \mu_{N,\RPN}^{\omega} )_{N \in \mathds N}$ $\mathds P$-almost surely is bounded and accordingly has at least one accumulation point. In addition, this shows that $\mathds P$-almost surely every accumulation point of $( \mu_{N,\RPN}^{\omega} )_{N \in \mathds N}$ is equal to or smaller than zero.
        \end{proof}

        \begin{lemma} \label{Lemma A.6}
            Let $(V_N)_{N \in \mathds N}$ be a Poisson random potential of fixed strength or with a strength that converges to infinity. If $\rho \ge \rho_{c,\RP_1}$, then $\lim_{N \to \infty} \mu_{N,\RPN}^{\omega} = 0$ $\mathds P$-almost surely.
        \end{lemma}

        \begin{proof}
            Let an $\omega \in \widetilde \Omega$ be arbitrarily given. According to Lemma~\ref{Lemma A.5}, $\limsup_{N \to \infty} \mu_{N,\RPN}^{\omega} \le 0$. Suppose $\liminf_{N \to \infty} \mu_{N,\RPN}^{\omega} < 0$. Then there is a subsequence $(N_m)_{m \in \mathds N}$ of $(N)_{N \in \mathds N}$ and a constant $\mu^{\omega}_{\infty} < 0$ such that $\lim_{m \to \infty} \mu_{N_m,\RP_{N_m}}^{\omega} = \mu^{\omega}_{\infty}$. Because we have $\rho = \int_{(0,\infty)} \mathcal B(E - \mu_{N, \RPN}^{\omega})  \, \mathcal N_{N, \RPN}^{\omega}(\mathrm{d} E)$ for all $N \in \mathds N$, by integration by parts, and by Lemma~\ref{Lemma A.1},
            \begin{align*}
                \rho & = \lim\limits_{m \to \infty} \lim\limits_{a \to 0} \lim\limits_{b \to \infty} \int\limits_{[a,b]} \mathcal B(E - \mu_{N_m, \RP_{N_m}}^{\omega})  \, \mathcal N_{N_m, \RP_{N_m}}^{\omega}(\mathrm{d} E)\\
                & \le \lim\limits_{m \to \infty} \lim\limits_{a \to 0} \lim\limits_{b \to \infty} \Big[ \int\limits_{[a,b]} \mathcal N_{N_m, \RP_{N_m}}^{\mathrm{I}, \omega}(E-) \, ( - \mathcal B'(E - \mu_{N_m, \RP_{N_m}}^{\omega})\, \mathrm{d} E \\
                & \qquad \qquad \qquad \qquad + \, \mathcal N_{N_m, \RP_{N_m}}^{\mathrm{I}, \omega}(b+) \mathcal B(b - \mu_{N_m, \RP_{N_m}}^{\omega})  \Big]\\
                & \le \lim\limits_{m \to \infty} \lim\limits_{a \to 0} \lim\limits_{b \to \infty} \Big[ \int\limits_{[a,b]} \mathcal N_{N_m, \RP_1}^{\mathrm{I}, \omega}(E-) \, ( - \mathcal B'(E - \mu_{N_m, \RP_{N_m}}^{\omega})\, \mathrm{d} E +  \mathcal N_{N_m, \RP_1}^{\mathrm{I}, \omega}(2b) \mathcal B(b/2) \Big]\\
                & \le \lim\limits_{m \to \infty} \lim\limits_{a \to 0} \lim\limits_{b \to \infty} \Big[ \int\limits_{[a,b]} \mathcal B(E - \mu_{N_m, \RP_{N_m}}^{\omega})  \, \mathcal N_{N_m, \RP_1}^{\omega}(\mathrm{d} E) + \mathcal N_{N_m, \RP_1}^{\mathrm{I}, \omega}(a-) \mathcal B(a - \mu_{N_m, \RP_{N_m}}^{\omega}) \Big]\\
                & \le \lim\limits_{m \to \infty} \int\limits_{(0,\infty)} \mathcal B(E - \mu_{\infty}^{\omega}/2)  \, \mathcal N_{N_m, \RP_1}^{\omega}(\mathrm{d} E) = \int\limits_{(0,\infty)} \mathcal B(E - \mu_{\infty}^{\omega}/2) \, \mathcal N_{\infty,\RP_1}(\mathrm{d} E) \\
                & < \int\limits_{(0,\infty)} \mathcal B(E) \, \mathcal N_{\infty,\RP_1}(\mathrm{d} E) = \rho_{c,\RP_1} \ .
            \end{align*}
        \end{proof}

        We finally prove the $\mathds P$-almost sure occurrence of g-BEC.
 
        \begin{theorem} \label{Theorem A.7}
            Let $(V_N)_{N \in \mathds N}$ be a Poisson random potential of fixed strength or with a strength that converges to infinity. Then we have
            \begin{align}
                \mathds P \Bigg( \lim\limits_{\varepsilon \searrow 0} \liminf\limits_{N \to \infty} \dfrac{1}{N} \sum\limits_{j \in \mathds N : E_{N,\RPN}^{j,\omega} \le \varepsilon} n_{N,\RPN}^{j,\omega} \ge \dfrac{\rho - \rho_{c,\RP_1}}{\rho} \Bigg) = 1 \ .
            \end{align}
            In particular, g-BEC $\mathds P$-almost surely occurs if $\rho>\rho_{c,\RP_1}$.
        \end{theorem}

        \begin{proof}
           Suppose $\rho \ge \rho_{c,V_1}$. Using Lemma~\ref{Lemma A.4} and Lemma~\ref{Lemma A.6} and proceeding similarly as in the proof of Proposition~\ref{Proposition 5.2}, we conclude
            \begin{align*}
                & \lim\limits_{\varepsilon \searrow 0} \limsup\limits_{N \to \infty} \int\limits\limits_{(\varepsilon,\infty)} \mathcal{B}(E - \mu_{N,\RPN}^{\omega}) \, \, \mathcal N_{N,\RPN}^{\omega}(\mathrm{d} E) \\
                & \quad \le \lim\limits_{\varepsilon \searrow 0} \Big[ \limsup\limits_{N \to \infty} \lim\limits_{M \to \infty} \int\limits_{[\varepsilon,M]} \mathcal N_{N,\RPN}^{\mathrm{I},\omega}(E-) \, ( - \mathcal{B'}(E - \mu_{N,\RPN}^{\omega})) \ud E  \\
                & \qquad \qquad + \, \limsup\limits_{N \to \infty} \lim\limits_{M \to \infty} \mathcal N_{N,\RPN}^{\mathrm{I},\omega}(M+) \, \mathcal{B}(M - \mu_{N,\RPN}^{\omega})\Big] \\
                & \quad \le \lim\limits_{\varepsilon \searrow 0} \Big[ \limsup\limits_{N \to \infty} \lim\limits_{M \to \infty} \int\limits_{[\varepsilon,M]} \mathcal N_{N,\RP_1}^{\mathrm{I},\omega}(E-) \, ( - \mathcal{B'}(E - \mu_{N,\RPN}^{\omega})) \ud E  \\
                & \qquad \qquad + \, \limsup\limits_{N \to \infty} \lim\limits_{M \to \infty} \mathcal N_{N,\RP_1}^{\mathrm{I},\omega}(2 M) \, \mathcal{B}(M - \mu_{N,\RPN}^{\omega})\Big] \\
                & \quad \le \lim\limits_{\varepsilon \searrow 0} \Big[ \limsup\limits_{N \to \infty} \lim\limits_{M \to \infty} \int\limits_{[\varepsilon,M]} \mathcal{B}(E - \mu_{N,\RPN}^{\omega}) \, \, \mathcal N_{N,\RP_1}^{\omega}(\mathrm{d} E) + \limsup\limits_{N \to \infty} \mathcal N_{N,\RP_1}^{\mathrm{I},\omega}(\varepsilon-) \, \mathcal{B}(\varepsilon - \mu_{N,\RPN}^{\omega})\Big] \\
                & \quad \le \lim\limits_{\varepsilon \searrow 0} \Big[ \limsup\limits_{N \to \infty} \int\limits\limits_{(\varepsilon,\infty)} \mathcal{B}(E - \mu_{N,\RPN}^{\omega}) \, \, \mathcal N_{N,\RP_1}^{\omega}(\mathrm{d} E) + 2 \mathcal{B}(\varepsilon/2) \limsup\limits_{N \to \infty} \mathcal N_{N,\RP_1}^{\mathrm{I},\omega}(2 \varepsilon) \Big]\\
                & \quad = \rho_{c,\RP_1}
            \end{align*}
            $\mathds P$-almost surely. Therefore, we $\mathds P$-almost surely have
            \begin{align*}
                & \lim\limits_{\varepsilon \searrow 0} \liminf\limits_{N \to \infty} \int\limits\limits_{(0,\varepsilon]} \mathcal{B}(E - \mu_{N,\RPN}^{\omega}) \, \, \mathcal N_{N,\RPN}^{\omega}(\mathrm{d} E) \\
                & \quad \ge \rho - \lim\limits_{\varepsilon \searrow 0} \limsup\limits_{N \to \infty} \int\limits\limits_{(\varepsilon,\infty)} \mathcal{B}(E - \mu_{N,\RPN}^{\omega}) \, \, \mathcal N_{N,\RPN}^{\omega}(\mathrm{d} E) \\\ 
                & \quad \ge \rho - \rho_{c,\RP_1} \ .
            \end{align*}
        \end{proof}

    {\small
	\bibliographystyle{amsalpha}
	\bibliography{mybibfile}} 

\end{document}